\documentclass[times, letter]{article}

\usepackage{geometry}

\usepackage{self}
\usepackage{amsthm}
\usepackage{times}
\usepackage{url}

\usepackage{amsmath}
\usepackage{longtable}
\usepackage{amsfonts}
\usepackage{amssymb}

\usepackage{graphicx}

\usepackage{multirow}
\usepackage{subcaption}
\usepackage{verbatim}
\usepackage{lineno}
\usepackage{enumitem}
\usepackage{color}
\usepackage{subcaption}
\usepackage{epsfig}
\usepackage{amsthm}

\usepackage{lineno}
\modulolinenumbers[5]

\usepackage{hyperref}
\usepackage{xcolor}
\hypersetup{
    colorlinks,
    linkcolor={red!50!black},
    citecolor={blue!50!black},
    urlcolor={blue!80!black}
}

\newtheorem{lemma}{Lemma}
\newtheorem{proposition}{Proposition}
\newtheorem{definition}{Definition}

\newtheorem{remark}{Remark}

\newcommand{\xx}{{\mathsf x}}
\newcommand{\XX}{{\mathsf X}}

\usepackage[numbers]{natbib}

\hypersetup{
    colorlinks,
    linkcolor={red!50!black},
    citecolor={blue!50!black},
    urlcolor={blue!80!black}
}

\graphicspath{{img/}}
\begin{document}

\title{Statistical inference of probabilistic origin-destination demand using day-to-day traffic data}

\author{Wei Ma, Zhen (Sean) Qian\\
Department of Civil and Environmental Engineering\\ Carnegie Mellon University, Pittsburgh, PA 15213\\
\{weima, seanqian\}@cmu.edu}
%\ead{weima@cmu.edu}
%\author[mymainaddress,mymainaddress2]{Zhen (Sean) Qian\corref{mycorrespondingauthor}}
%\ead{seanqian@cmu.edu}	
%\address[mymainaddress]{Department of Civil and Environmental Engineering\\Carnegie Mellon University, Pittsburgh, PA 15213}
%\address[mymainaddress2]{H. John Heinz III Heinz College\\Carnegie Mellon University, Pittsburgh, PA 15213}
\maketitle
\begin{abstract}
Recent transportation network studies on uncertainty and reliability call for modeling the probabilistic O-D demand and probabilistic network flow. Making the best use of day-to-day traffic data collected over many years, this paper develops a novel theoretical framework for estimating the mean and variance/covariance matrix of O-D demand considering the day-to-day variation induced by travelers' independent route choices. It also estimates the probability distributions of link/path flow and their travel cost where the variance stems from three sources, O-D demand, route choice and unknown errors. The framework estimates O-D demand mean and variance/covariance matrix iteratively, also known as iterative generalized least squares (IGLS) in statistics.  Lasso regularization is employed to obtain sparse covariance matrix for better interpretation and computational efficiency.  Though the probabilistic O-D estimation (ODE) works with a much larger solution space than the deterministic ODE, we show that its estimator for O-D demand mean is no worse than the best possible estimator by an error that reduces with the increase in sample size. The probabilistic ODE is examined on two small networks and two real-world large-scale networks. The solution converges quickly under the IGLS framework. In all those experiments, the results of the probabilistic ODE are compelling, satisfactory and computationally plausible. Lasso regularization on the covariance matrix estimation leans to underestimate most of variance/covariance entries. A proper Lasso penalty ensures a good trade-off between bias and variance of the estimation.

\end{abstract}

\section{Introduction}
Origin-destination (O-D) demand is a critical input to system modeling in transportation planning, operation and management. For decades, O-D demand is deterministically modeled, along with deterministic models of link/path flow and travel cost/time in classical traffic assignment problems. Transportation network uncertainty and reliability call for modeling the stochasticity of O-D demand, namely its spatio-temporal correlation and variation. With the increasing quantity and quality of traffic data collected years along, it is possible to learn the stochasticity of O-D demand for better understanding stochastic travel behavior and stochastic system performance metrics. Some studies considered the stochastic features of O-D demand, but few estimated the mean and variance of O-D demand from day-to-day traffic data. What is missing in the literature is the capacity to estimate spatially correlated multivariate O-D demand, in conjunction with a sound network flow theory on probabilistic route choices that can be learned from day-to-day traffic data. In view of this, this paper develops a novel data-friendly framework for estimating the mean and variance/covariance of O-D demand based on a generalized statistical network equilibrium. The statistical properties towards the estimated probabilistic O-D demand are also analyzed and provided. The process of O-D demand estimation (ODE) is further examined in real-world networks for insights.

O-D estimation (ODE) requires an underlying behavioral model, based on which O-D demand is estimated such that it best fits observations. Behavioral models in the static network context, namely route choice models, are also known as static traffic assignment models. The classical traffic assignment models \citep[e.g.,][]{dafermos1969traffic,patriksson1994traffic} deterministically map the deterministic O-D demand $q \in \mathbb{R}^{R\times S}_{+} $ to link flow $x \in \mathbb{R}^{N}_{+} $ or path flow $f \in \mathbb{R}^{K}_{+}$ (where $R$, $S$, $N$ and $K$ are the cardinality of sets of origins, destinations, links, and paths, respectively). In fact, the deterministic O-D demand $q$ is assumed to represent the \textit{mean} number of vehicles in the same peak hour from day to day. Likewise, link (path) flow is also deterministic, representing the \textit{mean} number of vehicles on a link (path) in the same hour from day to day. The classical traffic assignment models (such as User Equilibrium and Stochastic User Equilibrium) lay out the foundation of deterministic O-D estimation methods, namely to estimate $q$ in a way to best fit observed data related to a subset of link/path flow $x,f$. Deterministic O-D estimation (ODE) include the entropy maximizing models \citep{van1980most}, maximum likelihood models \citep{spiess1987maximum, watling1994maximum}, generalized least squares (GLS) models \citep{cascetta1984estimation, bell1991estimation, yang1992estimation, wang2016two}, Bayesian inference models \citep{maher1983inferences, tebaldi1998bayesian} and some recent emerging combined models \citep{ashok2002estimation, castillo2008trip, li2009markov, castillo2014hierarchical, lu2015enhanced, woo2016data}. For more details, readers are referred to the comprehensive reviews by \citet{bera2011estimation, castillo2015state}.

Classical traffic assignment models and ODE methods overlook the variance/covariance of demand and link/path flow, an essential feature for network traffic flow. Recent studies on network reliability and uncertainty model stochasticity of the network flow. Since ODE requires a traffic assignment model, we first review statistical traffic assignment models, followed by ODE models that take into account stochasticity.

One aspect of the stochastic network flow is using probability distributions to represent O-D demand. For example, \citet{waller2006network, duthie2011influence} sampled O-D demand from given multivariate normal distributions (MVN) and evaluated the network performance under classical User Equilibrium (UE) condition. \citet{chen2002travel} used a similar simulation-based method to evaluate travelers' risk-taking behavior due to probabilistic O-D demand. All these studies indicate the O-D variation is of great importance to network modeling and behavioral analysis. Statistical traffic assignment models consider various O-D probability distributions such as Poisson distributions \citep{clark2005modelling}, MVN \citep{castillo2008predicting}, multinomial distributions \citep{nakayama2016effect}. \citet{nakayama2014consistent} summarized different formulations and proposed a unified framework for stochastic modeling of traffic flows. Advantages and disadvantages of modeling traffic with those probabilistic distributions are also discussed by \citet{castillo2014probabilistic}. \citet{shao2006demand, shao2006reliability} proposed a reliability-based traffic assignment model (RUE) and extended it to consider different travelers' risk taking behavior. \citet{lam2008modeling} further extended the model to consider the traffic uncertainty and proposed reliability-based statistical traffic equilibrium. \citet{zhou2008comparative, chen2010alpha} proposed a $\alpha$-reliable mean-excess traffic assignment model which explicitly models the travel time distribution and consider the reliability and uncertainty of the travel time on travelers' route choice behavior. Other studies \citep{haas1999modeling, waller2001evaluation, duthie2011influence} show that the variance/covariance matrix of the O-D demand have a significant influence on network traffic conditions.

Though adopting stochastic O-D demand, those traffic assignment models (except \citet{clark2005modelling}) assumed non-atomic (infinitesimal) players and therefore are unable to capture the stochasticity of route choices that vary from day to day (the proof is shown in \citet{GESTA}). Classical UE, SUE and RUE are all deterministic route choice models where the number of (infinitesimal) players assigned to each route is fixed, rather than being stochastic. Thus, those models are unable to explain the day-to-day variation of observed traffic counts at the same location and the same time of day. To further see how classical equilibrium models overlook the route choice stochasticity, suppose there is $Q$ travelers where $Q$ is a random variable to be realized on each day. Given the probability of choosing a route $p$, the route flow is deterministically identified by the number of infinitesimal users who take this route, $F = pQ$. Even if the route choice probability $p$ is determined by stochastic choice models (such as Probit, Logit, etc.), these models still assume that a fixed number $Qp$ of travelers take this route on each day, as a result of non-atomic equilibrium. This does not, theoretically, allow to model the day-to-day variations of travelers' choices. Recent studies on statistical traffic assignment models indicate that the route flow is the aggregation of random choices of O-D demand, and thus also random \citep{watling2002second, watling2002second2, Shoichiro2006STOCHASTIC, nakayama2003traffic, nakayama2014consistent}. Travelers' route choice follows a multinomial distribution with the probability obtained from route choice models, $f \sim \text{Multinomial}(Q, p)$. To distinguish to what extent stochasticity is modeled for route choices, we refer to the former (classical) models as ``\textit{fixed portions with stochastic route choice models }" and the latter models as ``\textit{probabilistic distributions with stochastic route choice models}". Though route choices are stochastic, these studies did not work directly with the covariance of demand among all O-D pairs. A detailed comparison of those assignment models is further illustrated in \citet{GESTA}.

%Three important features to model the stochastic road networks can be extracted from the above statistical traffic assignment models. These features are essential for statistical traffic assignment models to capture the stochastic traffic conditions on more complex  and larger congested road networks. The first feature is modeling the O-D demand as a correlated multivariate random vector.

%So viewing the O-D demand as correlated multivariate random vector captures more network stochascity than simply viewing it as constants or uncorrelated univariate random variables.
Given a route choice probability, we can derive probability distributions of path/link flow. However, it is non-trivial to establish a statistical network equilibrium where the route choice probability is determined endogenously as a result of stochastic O-D demand, path/link flow and network conditions. In the deterministic settings, UE, SUE or RUE simultaneously determines the \textit{mean} path/link flow, and the route choice probability \citep[e.g.,][]{yang1992estimation}. Very few studies examined the statistical network equilibrium. \citet{davis1993large} proposed a Markov process to model the day-to-day variation of traffic flow given a fixed regional population. The stochastic route choice on a particular day is assumed to be related to the stochastic network condition of the previous days \citep{cascetta1991day}. When the network evolves from day to day, there exists a network equilibrium where stabilized probability distributions of path/link flow are reached. Different from this approach, \citet{GESTA} proposed a generalized statistical equilibrium where each traveler makes stochastic choices based on his/her entire past experience, namely the probability distributions of the equilibrated network conditions. \citet{GESTA} integrates multivariate probability distributions of O-D demands and link/path flow into the stochastic route choice models, and ultimately solved for the probabilistic network flow. It also analytically decomposes the variance of link/path flow into three sources, O-D demand variation, route choice variation and measurement error.

%Later we'll show that the ignorance of travelers' route choice variation in O-D estimation method will lead to an overestimate of O-D variation. The influence will be particularly significant when several routes share similar costs as their route choice probability are similar, hence the route choice variation is even larger.

%The second feature is considering the travelers' route choice endogenously under congested traffic conditions . Rather than assuming the travelers' route choices are known or observable,  it should be inferred by the travelers' equilibrium condition. This feature enables the traffic assignment models to apply on congested networks. \citet{di2016second} discussed the network toll policy under two extreme travelers risk-taking behavior (risk-averse and risk-prone).

%Travelers' route choice variation is the third indispensable feature of statistical traffic assignment. Our previous researches on statistical traffic assignment model show that the travelers' route choice variation takes approximately 20\% of the total variation.
%In behavior science, the variation induced by route choice uncertainty is under intensive exploration.

With little work on statistical traffic assignment models, probabilistic ODE can be challenging. An ideal probabilistic ODE should posses three critical features: estimating variance/covariance of O-D demand in addition to its mean, a statistical network equilibrium that can fit massive data collected years along, and consideration of day-to-day route choice variation. To our best knowledge, few literature works with ODE models that take into account any of the three features. We summarize existing representative ODE models in Figure~\ref{fig:lit}.

%As a reverse process of statistical traffic assignment models, the probabilistic O-D demand estimation methods should inherit the three critical features from the statistical traffic assignment models. That is to say, a probabilistic O-D demand estimation method should estimate the correlated multivariate O-D demand vector, and consider the equilibrium condition and route choice variation at the same time.

\begin{figure}[h]
\centering
\includegraphics[width=0.7\textwidth]{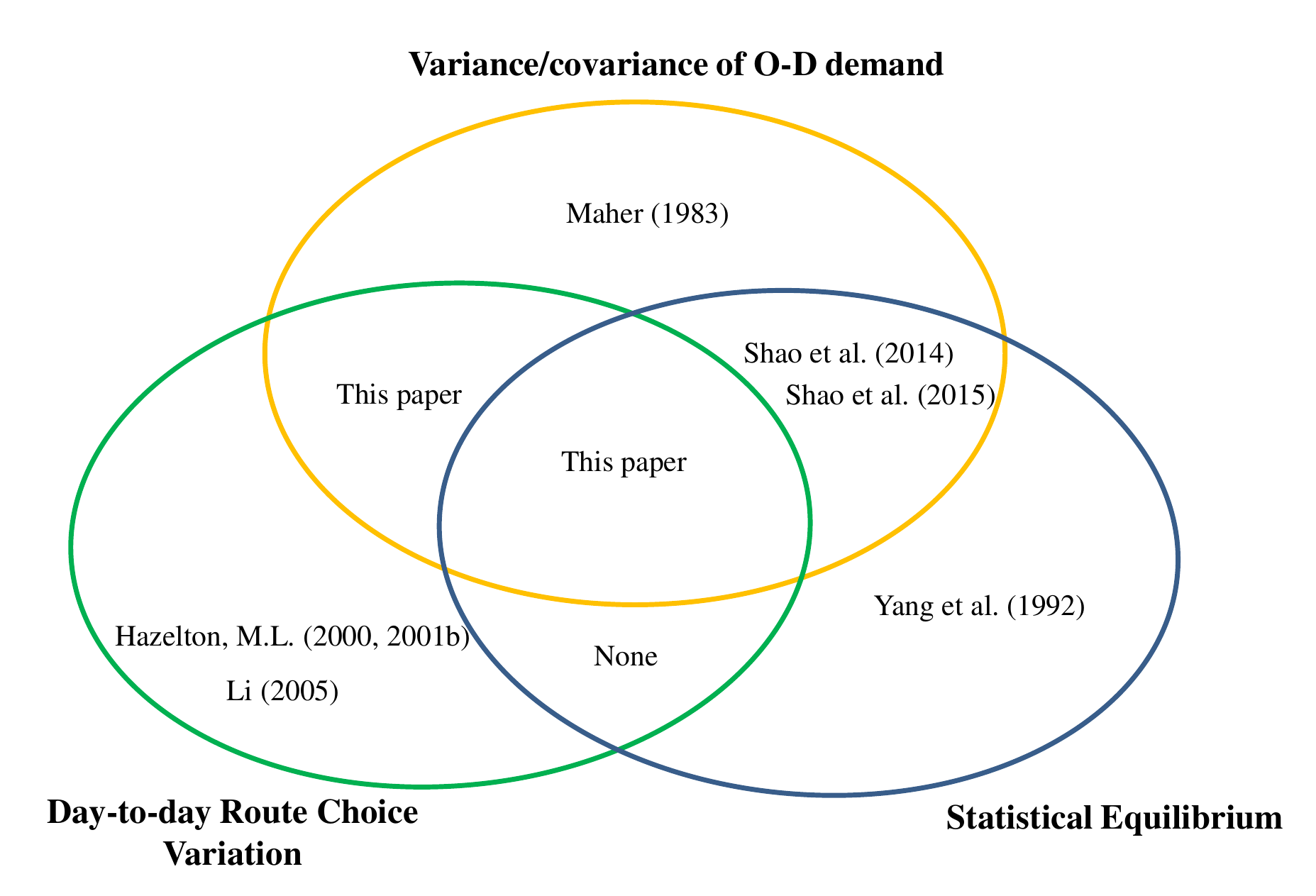}
\caption{\footnotesize O-D estimation methods by categories}
\label{fig:lit}
\end{figure}

To further distinguish our work from existing literature, \citet{vardi1996network, hazelton2000estimation, hazelton2001estimation, hazelton2001inference, li2005bayesian, parry2012estimation} assumed O-D demand follows Poisson distributions that are independent among O-D pairs, and \citet{hazelton2000estimation, hazelton2001estimation, hazelton2001inference, li2005bayesian, parry2012estimation} further considered day-to-day route choice variation. They formulated maximum likelihood estimator and Bayesian inference method for O-D demand, whereas the demand covariance among O-D pairs is not considered. In addition, using uncongested networks simplifies the route choice model, and thus no network equilibrium under congestion is proposed. \citet{shao2014estimation, yang2017stochastic} proposed a generalized model to estimate the mean and variance of O-D demand with MVN. \citet{shao2015estimation} further extended the model to estimate multi-class O-D demand and used $L^1$ regularizer to shrink the model dimensions.  An equilibrium is used to model deterministic route choices (similar to classical UE and RUE), while neither of them considered day-to-day route choice variation.

%This paper builds an ODE model based on the generalized statistical traffic assignment model proposed by \citet{GESTA}.

To our best knowledge, there is a lack of study that estimates correlated multivariate probabilistic O-D demand under a mathematically sound statistical network equilibrium (namely a truly stochastic route choice model) while considering day-to-day random route choices simultaneously. To fill up this gap, this paper builds an ODE model based on the generalized statistical traffic assignment model (GESTA) proposed by \citet{GESTA}. Any observation (such as link/path flow, or travel time/cost) has a variance that stems from three sources, O-D demand variation, route choice variation and unknown errors, all from day to day. Using GESTA as the underlying behavioral model, we estimate probability distributions of O-D demand using data from various data sources. Furthermore, conventional goodness-of-fit indicators \citep[e.g.,][]{antoniou2015towards} is not suitable for probabilistic ODE. This paper also proposes new goodness-of-fit indicators based on probability distributions to evaluate the performance of probabilistic ODE.
%With the estimated probabilistic O-D demand, a hypothesis testing framework can be built based on multi-day data. \citet{zhou2003dynamic} built an univariate F-test for different weekdays. However, since the link flow vector is highly correlated, univariate analysis is no longer a good choice. In this study we build the hypothesis testing based on the distribution of multivariate link/path flow.

Another important issue for ODE is the observability problem \citep{castillo2008observability, yang2015data}. It is well known that the ODE is underdetermined using observed link-based traffic counts \citep{yang1992estimation}. \citet{hazelton2003some} suggested to use second-order statistical information of the observed link counts to estimate Poisson distributed O-D demand without covariance among O-D pairs. By utilizing the second-order information, the Poisson distributed O-D demand can be estimated uniquely. \citet{yang2015data} proposed to uniquely determine the O-D demand by properly selecting data among observed link flow data and historical O-D information. \citet{yang2017stochastic} also showed that possibly dissimilar multi-day observation improves the observability of OD demand \citep{yang2017stochastic}.  \citet{hazelton2015network} proposed new method for sampling
latent route flows conditional on the observed link counts when the network’s link-path
incidence matrix is unimodular. In this paper, the observability of the proposed probabilistic ODE is examined. The new ODE guarantees that its estimator for the O-D demand mean is no worse than the conventional deterministic ODE by the order of ${\cal O}(\frac{1}{n})$, provided with sufficient data.

Under the proposed probabilistic ODE framework, we propose to estimate the O-D demand mean and variance-covariance matrix iteratively, which decomposes the complex ODE into two sub-problems. In the sub-problem of estimating O-D demand mean vector $q$, we extend the statistical ODE from \citet{menon2015fine}, and both single level and bi-level ODE formulations are discussed. In the sub-problem of estimating O-D demand variance-covariance matrix $\Sigma_q$, we utilize estimated link flow covariance to formulate the estimation problem, and then apply Lasso regularization and convex relaxation on the formulation. How to use traffic speed data in addition to traffic counts data is discussed. Furthermore, the statistical properties of the estimated probabilistic O-D demand are provided for insights. The observability of the ODE problem is also examined.

The main contributions of this paper are summarized as follows:
\begin{enumerate}[label=\arabic*)]
\item It proposes a novel theoretical framework for estimating probabilistic O-D demand (namely mean and variance/covariance of O-D demand) considering newly defined generalized statistical network equilibrium. The statistical equilibrium simultaneously integrates probabilistic O-D demand and travelers' day-to-day route choice variation.
\item It defines the goodness-of-fit for probabilistic ODE and develops a theory for variance analysis of estimated probabilistic O-D demand and path/link flow.
\item It discusses the observability issue for probabilistic O-D demand, and shows that the estimated mean of O-D demand using the the probabilistic ODE is no worse than the estimate O-D demand using a deterministic ODE model, provided with sufficient data.
\item It intensively examines the proposed probabilistic ODE framework on two large-scale real networks using both simulated traffic data and real world traffic counts to gain insights from solutions, as well as to show the computational efficiency of the solution algorithms.
\end{enumerate}

The remainder of this paper is organized as follows. Section \ref{sec:example} presents an illustrative example to further explain how and why it is necessary to consider day-to-day route choice variation. Section \ref{sec:formulation} discusses the formulation details, followed by section~\ref{sec:analysis} presenting properties of the model and addressing the observability issues. Section \ref{sec:soluion} proposes the entire probabilistic ODE framework. In section~\ref{sec:exp}, two simple illustrative examples are used to demonstrate the concepts and ODE results. Two large scale networks are used to demonstrate the computational efficiency of the probabilistic ODE method, and its ability to work with real-world data. Finally, conclusions are drawn in section \ref{sec:con}.

\section{An illustrative example}
\label{sec:example}
Observations of link/path flow varies from day to day. In principle, the day-to-day variation is attributed to three sources, O-D demand variation, route choice variation and sensing measurement error. The impact of day-to-day O-D demand variation and measurement error on link/path flow have been thoroughly discussed by \citet{yang1992estimation, waller2001evaluation, shao2014estimation}. Here we use a toy example to compare probabilistic O-D estimation results with and without the consideration of day-to-day route choice variation. It is intended to illustrate that where the day-to-day route choice variation comes from, and why it is important to not neglect it when estimating O-D demand.

Consider a toy network as shown in Figure \ref{fig:net1}, on each day $Q$ amount of vehicles depart from node $r$ to node $s$. We assume $Q$ is normally distributed, so is the link flow (i.e. path flow in this example) $X_1$ and $X_2$, the number of vehicles on link 1 and link 2, respectively. Suppose both links are indifferent. We do not consider measurement error here.
\begin{figure}[h]
\centering
\includegraphics[width=0.4\textwidth]{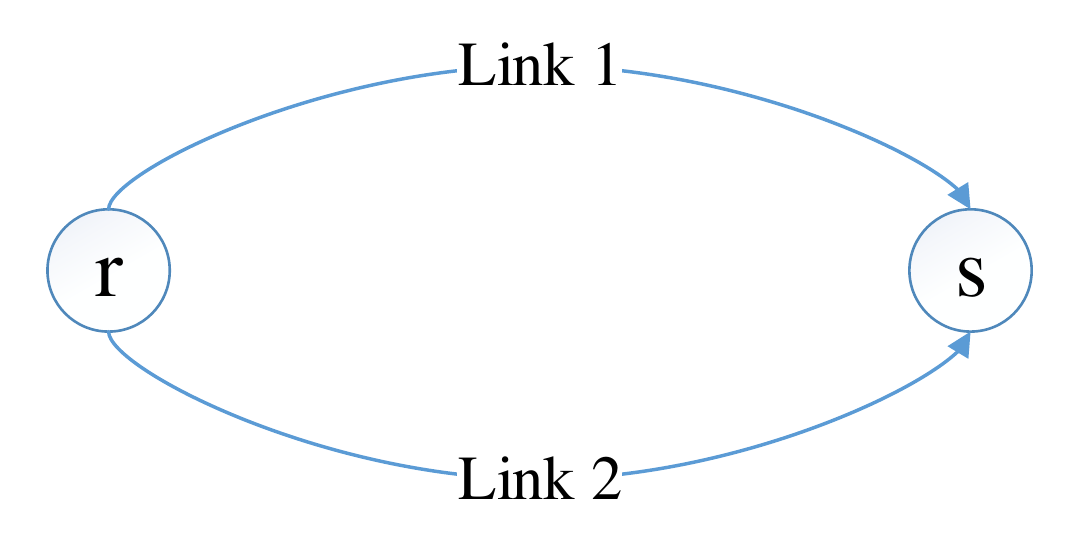}
\caption{\footnotesize A toy network}
\label{fig:net1}
\end{figure}

On each day, the demand $Q$ is realized, but cannot be directly measured. Instead, we measure the link flow $X_1$ on each day. Suppose after observations of many days, we determine the probability distribution of $X_1 \sim \N(50, 10^2)$.  If we do not consider the day-to-day route choice variation, the estimated O-D demand follows the rounded normal distribution $Q = \left[Q'\right], Q' \sim \N(100, 20^2)$ since $Q \simeq 2X_1$ by equilibrium conditions. In fact, regardless of which route choice models (logit, probit, etc.), the distribution of the O-D demand can be directly computed since both links are indifferent to travelers. Both links have the same flow distributions.

Now consider travelers' route choice that varies from day to day. The probability of any vehicle choosing link $1$ and link $2$ is $p_1 =p_2= 0.5$. All $Q$ vehicles make this choice independently on each day. Thus, the link flow follows a multinomial distribution (binomial distribution in this example),
\begin{eqnarray}
(X_1, X_2)^T &\sim& \text{Multinomial}(\left[Q\right], (0.5, 0.5)^T)\\
Q &\sim& \N(q, \Var(Q))
\end{eqnarray}
Where $q$ denotes the mean of the probability distribution of $Q$. Since $q= 2\Expect(X_1) = 100$, by law of total variance we have,
\begin{eqnarray}
\Var(X_1) &=& \Var(\mathbb{E}(X_1|Q)) + \mathbb{E}( \Var(X_1|Q)) = p_1^2 \Var(Q)+ p_1p_2\mathbb{E}(Q)\\
100&=& 0.25 \Var(Q) + 0.25 \mathbb{E}(Q)\\
100&=&  0.25 \Var(Q) + 25
\end{eqnarray}
Thus $\Var(Q) = 300$. Clearly, the probabilistic O-D estimation without considering day-to-day route choice variation overestimates the variance of O-D demand by the term $p_1p_2\mathbb{E}(Q)$, namely $33.3\%$, a fairly substantial quantity. Intuitively, probabilistic route choices that are made by travelers independently on a daily basis tend to reduce the day-to-day variation of route flow. As shown in \citet{GESTA}, classical Stochastic User Equilibrium theory works with non-atomic (infinitesimal) users, which in theory does not allow day-to-day choice variation. This greatly simplifies the computation. However, it potentially conflicts with real-world observations (such as counts and speeds) that vary from day to day, and thus might be difficult to handle massive data collected over many years. Provided that the route choice variation can be substantial in real-world traffic, our ODE method needs to model atomic users, and take day-to-day route choice variation into consideration while estimating probability distributions of O-D demand.

%In the following section, we develop a generalized probabilistic ODE framework and compare it with the existing probabilistic ODE methods.

\section{Formulations}
\label{sec:formulation}
In this section, we discuss the probabilistic ODE framework. We first present the notations and assumptions, compare this framework to existing formulations, and finally discuss each component of the framework in details.

\subsection{Notations}
\label{sec:notation}
Please refer to Table~\ref{tab:notation}. The superscript $\cdot^o$ indicates that the variable is projected onto observed links only. The hat symbol, $\hat{\cdot}$, indicates the variable is an estimator for the true (unknown) variable.

\begin{longtable}{p{3cm}p{10cm}}
\caption{\footnotesize List of notations}
\label{tab:notation}
\endfirsthead
\endhead	
\multicolumn{2}{c}{\textbf{Network Variables}}\\
$A$ & The set of all links\\
$A^o$ & The set of links with flow observations\\
$N$ & The number of links\\
$K_q$ & The set of all O-D pairs\\
$K_{rs}$ & The set of all paths between O-D pair $r-s$\\
$\Delta$ & Path/link incidence matrix\\
$\Delta^o$ & Path/observed link incidence matrix \\
$M$ & Path/O-D demand incidence matrix \\
%$n_{rs}$ & The number of paths between O-D pair $r-s$\\

\multicolumn{2}{c}{}\\
\multicolumn{2}{c}{\textbf{Random Variables}}\\
$Q_{rs}$ & The demand of O-D pair $r-s$\\
$Q$  & The vector of O-D demands\\
$X_a$ & Link flow on link $a$\\
%$Q^H$ & History O-D demand vector\\
$X$  & The vector of link flow\\
$X^o$  & The vector of link flow on observed links\\
$X_m$  & The vector of measured link flow\\
$F_{rs}^k$ & Path flow on path $k$ between O-D pair $r-s$\\
$F_{rs}$  & The vector of path flow between O-D pair $r-s$\\
$F$  & The vector of path flow between all O-D pairs\\
$C$ & The vector of path costs between all O-D pairs\\
%$C_{rs}^k $ & The travel costs on path $k$ between O-D pair $r-s$\\
$E$ & The vector of unknown error for links\\
\multicolumn{2}{c}{}\\
\multicolumn{2}{c}{\textbf{Parameters for Probability Distributions}}\\
$q=\mathbb{E}(Q)$  & The vector of means of O-D demands \\
$q_{rs}=\mathbb{E}(Q_{rs})$ & The mean of O-D demand $Q_{rs}$\\
$\Sigma_q$  & Covariance matrix of O-D demands\\
$p_{rs}^k$ & The probability of choosing path $k$ in all paths between O-D pair $r-s$\\
$p_{rs}$ & The vector of route choice probabilities of all paths between O-D pair $r-s$\\
$p$  & Route choice probability matrix, consisting of all $p_{rs}^k$\\
$c=\mathbb{E}(C)$ & The vector of means of path costs for all O-D pairs\\
$\Sigma_c$	& Covariance matrix of path costs vector\\
%$M$  & path flow to O-D matrix\\
$x=\mathbb{E}(X)$  & The vector of means of link flow\\
$x_a=\mathbb{E}(X_a)$ & The mean of link flow $X_a$\\
$\Sigma_x$  & Covariance matrix of link flow\\
$f=\mathbb{E}(F)$  & The vector of means of path flow\\
$f_{rs}^k=\mathbb{E}(F_{rs}^k)$ & The mean of path flow $F_{rs}^k$\\
$\Sigma_f$  & Covariance matrix of path flow\\
$\Sigma_e$  & Covariance matrix of unknown error\\

\multicolumn{2}{c}{}\\
\multicolumn{2}{c}{\textbf{Parameters for Conditional Probability Distributions}}\\
$\Sigma_{f|Q}$ & Covariance matrix of path flow conditional on O-D demands $Q$\\
%$\Sigma_{f|q}$ & Covariance matrix of link flow conditional on the mean of $Q$\\
$\Sigma_{f|q}$ & Covariance matrix of path flow conditional on O-D demands $Q=q$, $\Sigma_{f|q} = \Sigma_{f|Q = q}$ \\

\multicolumn{2}{c}{}\\
\multicolumn{2}{c}{\textbf{Observed Link Flow Data}}\\
$\xx^o$ & Observed link flow matrix\\
$n$	 & Number of observed link flow vectors, namely number of observed days\\
$\xx_i^o$  & The $i$-th observed link flow vector (on the $i$-th day), $1 \leq i \leq n$\\
$S_x^o$   & Empirical covariance matrix of the observed link flow
\end{longtable}

\subsection{Assumptions}
\begin{enumerate}
\item O-D demands follow a rounded multivariate normal (MVN) distribution with mean $q$ and covariance matrix $\Sigma_q$.
\begin{eqnarray}
Q = \left[ Q' \right],  Q' \sim \N(q, \Sigma_q)
\end{eqnarray}
where the random variable $Q$ are generated from a standard continuous MVN $Q'$ and then rounded to the nearest integer. When the number of travelers is sufficiently large, the rounding error is negligible \citep{nakayama2016effect}. In this paper demand is  approximated by a continuous MVN. Because the demand is always non-negative, the MVN is truncated at zero. When the O-D demands are sufficiently large, the effect of the truncation is also negligible.

\item All travelers between the same O-D pair are homogeneous when making route choices. They perceive the probability distribution of travel time/cost of the network from their entire past experience, similar to \citet{nakayama2014consistent}. On each day, each atomic traveler independently and identically makes a route choice.

\item Each traveler makes her route choice decision solely based on the perception of the past traffic conditions, unlike the classical Wardrop UE or atomic Nash Equilibrium where each traveler is fully aware of others' choices. %The stabilized probability distributions of network conditions, in turn, results in stabilized probability distributions of individual route choices.

\item The day-to-day variation of observed path/link traffic flow is resulted from O-D demand variation, travelers' route choice variation and unknown errors. Unknown errors include measurement errors, and other unobserved error (such as traffic incident and non-recurrent traffic behavior). Unknown errors are link-based and have zero mean, which also follow MVN.
\begin{eqnarray}
E \sim \N(\vec{0}, \Sigma_e)
\end{eqnarray}

\item The observed link flow vector $\xx_i^o$ on the $i$-th day ($1 \leq i \leq n$) is an i.i.d sample from the probability distribution of the observed link flow vector $X_m^o$.

\end{enumerate}

%\subsection{Review of Atomic Statistical Equilibrium}

\subsection{Review of the generalized statistical traffic assignment (GESTA) model }

As the underlying model, a generalized statistical traffic assignment model (GESTA) is proposed in \citet{GESTA}. We briefly review the GESTA, and then propose the ODE framework based on GESTA.

Consider a graph with $|N|$ nodes, $|A|$ links, $|K_q|$ O-D pairs and $|K_{rs}|$ paths for each O-D pair $r-s$, GESTA maps O-D demands $Q$ to link/path flows $X, F$ and path costs $C$ under a \textbf{statistical equilibrium}. From day to day, the recurrent traffic reaches a \textbf{statistical equilibrium} among $Q,C,F,X$ such that, provided with exogenous demand $Q$ (a random variable vector), $C,F,X$ follow stabilized probability distributions that can be represented by a stabilized route choice probability vector $p$ \citep{GESTA}. The \textbf{statistical equilibrium} is defined as follows,

\begin{definition} \citep{GESTA}
A transportation network is under a statistical equilibrium, if all travellers practice the following behavior: on each day, each traveler from origin $r$ to destination $s$ independently chooses a route $k$ with a deterministic probability $p_{rs}^k$. For a sufficient number of days, this choice behavior leads to a stabilized probability distribution of travel costs with parameters $\vartheta$. This stabilized probability distribution, in turn, results in the probabilities $p=\psi(\vartheta)$ where $\psi(\cdot)$ is a general route choice function.
\end{definition}

Note that the definition of statistical equilibrium is independent of time and describes the equilibrated condition. The definition indicates that the traffic flow/costs on each day are random and by no means identical from day to day. However, the probability distributions of traffic conditions over a course of a number of days (such as months or years) are stabilized. Based on the O-D demand $Q$, we express the probability distributions of $X, F, C$ all by the deterministic probability vector $p$, and construct a fixed point problem regarding $p$ to solve for $X, F, C$.

Note that the route choice model $\psi(\vartheta)$ can be generic. For instance, $\psi(\vartheta)$ can represent the classical UE where $\vartheta\doteq\mathbb{E}(C) = c$. If we use a random utility model as the route choice model, then the route choice probability becomes
\begin{equation}
p_{rs}^k = \Pr\left[C_{rs}^k \leq \min_{i \neq k} C_{rs}^i\right] \label{eq:choice}\\
%&\simeq& \Phi\left[\frac{\Expect C_{n} - \Expect C_{-n}}{\sqrt{\Var(C_n) + \Var(C_{-n}) - 2 \Cov(C_n, C_{-n})}} \right]
\end{equation}
Each traveler perceives his/her stochastic costs on all possible paths between O-D pair $r-s$, and the probability of choosing path $k$ equals the probability of path costs outperforming costs of other paths. For this reason, $\psi(\vartheta)$ can also easily represent the Logit model and Probit model. However, Probit model would be the most convenient and natural representation. This is because $C$ is approximated by MVN under GESTA, and this is consistent with the normally distributed perceived travel cost under Probit. If the Logit model is adopted for $\psi(\vartheta)$, then one needs to assume that travelers' perception on travel time/cost would add a term that follows Gumbel distribution in addition to the cost mean $\mathbb{E}(C) = c$.

The hierarchical formulation of GESTA is represented in Equation~\ref{eq:multi}.
\begin{equation}
\begin{aligned}
\label{eq:multi}
Level~1: && X_m = X+ \epsilon_e  && \text{(Unknown Error)}\\
&&\epsilon_e \sim \N(\vec{0}, \Sigma_e)\\
Level~2: && X = \Delta F\\
&& F \sim \mathcal{MN}(\tilde{p}Q, \Sigma_f)  && \text{(Route choice variation)}\\
Level~3: && Q \sim \mathcal{N}(q, \Sigma_q)  && \text{(Demend variation)}
\end{aligned}
\end{equation}

Level $3$ represents the O-D demand variation. The Level $2$ formulation indicates that each traveler makes their route choice independently, thus the path flow vector $F$ follows a multinomial distribution. In Level $1$, the unknown errors are applied to the link flow vector. $X_m$ is the random variable implying the observations of link flow.

Define $\tilde{p}\doteq diag(p)B $ and $B$ is a transition matrix (with blocks of ones and zeros) defined in \citet{watling2002second,GESTA}. After some derivations, the path/link flow distributions can be obtained, which are stated in the following two propositions.

\begin{proposition}
\label{pro:gesta1}
The marginal distribution of $F$ can be approximated by,
$$F \sim \N (f, \Sigma_f) $$
where $f =\tilde{p} q$, $\Sigma_f =\Sigma_{f|q} + \tilde{p}\Sigma_{q}\tilde{p}^T$, $\Sigma_{f|q}$ is the covariance matrix of path flow conditional on O-D demands $Q=q$ in multinomial distribution, $\Sigma_{f|q} = \Sigma_{f|Q = q}$. The matrix $\Sigma_{f|Q}$ is built upon $p$, which can be found in \citet{GESTA}.
\end{proposition}

\begin{proposition}
\label{pro:marginal}
The marginal distribution of $X$ and $X_m$ follows,
$$X \sim \N (x, \Sigma_x)$$
$$X_m \sim \N (x, \Sigma_x+ \Sigma_e)$$
where $x = \Delta f$, $\Sigma_x = \Delta \Sigma_{f} \Delta^T  $.
\end{proposition}

\subsection{Review of existing probabilistic ODE formulations}
\label{sec:rev}
%Using GESTA as the fundamental behavior model, we can develop the estimation framework for probabilistic O-D demand.
In this subsection, we will review existing probabilistic ODE formulations and discuss the research gap. The maximum likelihood formulations proposed by \citet{vardi1996network, hazelton2000estimation} rely on the assumption that the O-D demand follows Poisson distributions that are independent among O-D pairs. The Poisson distributed assumption is actually quite restrictive since it explicitly adds a constraint between the mean and variance, often violated in practice \citep{wang2016two}. More importantly, the maximum likelihood formulation cannot be easily extended to the case of multivariate probability distributions for O-D demand, since the likelihood function will become intractable.

\citet{shao2014estimation, shao2015estimation} proposed a generalized least square formulation to estimate the mean and variance of O-D demand that follows the MVN. The formulation is a bi-level optimization problem. The upper level minimizes the difference between observed and estimated link flow mean/variance. The objective function can be written as,

\begin{equation}
\begin{array}{rrclcl}
\vspace{5pt}
\label{eq:shao}
\displaystyle \min_{q, \Sigma_q} & \multicolumn{3}{l}{\displaystyle \alpha_1f_1(x(q), x^o) + \alpha_2f_2(\Sigma_x(\Sigma_q), \Sigma_x^o)}\\
\end{array}
\end{equation}
where $\alpha_1$ and $\alpha_2$ are constant weights and $f_1(\cdot, \cdot), f_2(\cdot, \cdot)$ measures the difference between the observed and estimated link flow mean and covariance/variance, respectively. $x(q)$ implies that the mean of link flow is a function of O-D demand mean, whereas $\Sigma_x(\Sigma_q)$ implies that the variance/covariance matrix of link flow is a function of O-D demand variance/covariance matrix.

As discussed in the introduction section, Formulation~\ref{eq:shao} does not consider the route choice stochasticity (namely the day-to-day route choice variation) as a result of non-atomic equilibrium. Thus, $\Sigma_x$ is only dependent on the O-D variation and physical properties of the network. The derivative of the objective function can be derived analytically. Formulation~\ref{eq:shao} can be solved using gradient methods. However, if the route choice stochasticity is considered in link flow for atomic users, then $\Sigma_x$ is dependent on both O-D demand mean $q$ and variance $\Sigma_q$ by Propositions~\ref{pro:gesta1} and~\ref{pro:marginal}, the objective function becomes,
\begin{equation}
\begin{array}{rrclcl}
\vspace{5pt}
\displaystyle \min_{q, \Sigma_q} & \multicolumn{3}{l}{\displaystyle \alpha_1f_1(x(q), x^o) + \alpha_2f_2(\Sigma_x(q, \Sigma_q), \Sigma_x^o)}\\
\end{array}
\end{equation}
where both $x(q)$ and $\Sigma_x(q, \Sigma_q)$ are determined by GESTA \citep{GESTA}. Since the function $\Sigma_x(q, \Sigma_q)$ is complex, deriving its gradient is  challenging. Though gradient-free algorithms can be used to find the optimal solution, those methods can have a hard time being applied to large-scale networks. To effectively solve the new objective function subject to the statistical equilibrium based GESTA, this paper proposes the following iterative method.

\subsection{A novel iterative formulation for estimating probabilistic O-D demand}
Instead of estimating the mean and variance/covariance matrix of O-D demand simultaneously, we propose an Iterative Generalized Least Square (IGLS) framework to estimate mean and variance/covariance iteratively \citep{goldstein1986multilevel}. A large number of statistical algorithms are developed for IGLS. \citet{del1989unifying} shows that IGLS shares some properties with the Newton-Raphson algorithm, and IGLS can be used for solving Maximum Likelihood Estimator (MLE) and quasi-MLE problems.

IGLS framework mainly contains two sub-problems: estimating the O-D demand mean and estimating the variance/covariance matrix. IGLS runs the two sub-problems iteratively to update estimators $\hat{x}, \hat{f}, \hat{p}, \hat{q}, \hat{\Sigma}_q$. In the sub-problem of estimating the mean vector $q$, $\hat{\Sigma}_q$ and $\hat{\Sigma}_x$ are seen as given.  In the sub-problem of estimating the variance-covariance matrix $\Sigma_q$, $\hat{q}$ is seen as given. As we will show later, the link flow covariance $\Sigma_x$ follows the Wishart distribution. We then formulate the ODE as an approximated maximum likelihood estimator of the covariance matrix with Lasso regularization. %The detailed formulations for each sub-problem will be presented in the following sections, and the step by step description of IGLS will be presented in Section~\ref{sec:soluion}.

Compared to the framework by \citet{shao2014estimation, shao2015estimation}, we argue that the IGLS-based formulation is statistical interpretable and more computationally efficient to solve on large-scale networks. As shown in \citet{del1989unifying}, each iteration in the IGLS resembles one gradient descent step in the Newton-Raphson algorithm. Since the convergence rate of the Newton-Raphson algorithm is quadratic, a good convergence rate can be expected for IGLS. Later, we also show that the IGLS framework can theoretically provide a better performance in terms of the solution algorithm.

Before detailed models for each of the two sub-problems are presented, we first discuss the stopping criterion. Stoping criterion is needed for three processes in the IGLS framework: estimating mean, estimating O-D variance/covariance matrix and the overall IGLS iteration. For the two sub-problems, since both optimization problems are convex (as we will show later), the stopping criterion can be easily defined \citep{boyd2004convex}. For the stopping critera for the overall IGLS iteration, note that our ultimate goal is to estimate the probability distribution of the O-D demand $Q$.  We keep track of the discrepancy between the mean and variance/covariance matrix along the iterations. If the discrepancy is sufficiently small, then the IGLS iteration can be stopped.

\begin{definition}[Stopping Criterion for IGLS]
If the probability distribution of $Q$ is estimated as $\N(\hat{q}^+, \hat{\Sigma_q}^+)$ following the estimator from a previous iteration $\N(q, \Sigma_q)$, define
\begin{eqnarray}
\tau &=& D((\hat{q}^+, \hat{\Sigma}_q^+)^T, (\hat{q}, \hat{\Sigma}_q)^T)
\end{eqnarray}
where $D(\cdot, \cdot)$ is a distance measure between the two estimators. The IGLS iteration terminates when $\tau$ is sufficiently small.
\end{definition}

There are a large number of candidates for the choice of $D(\cdot, \cdot)$. In this paper we choose two distance functions: Hellinger distance and Kullback-Leibler distance (KL distance), computed as the following for any two MVNs $\N(\mu_1, \Sigma_1), \N(\mu_2, \Sigma_2)$ with the same dimension $d$,
\begin{eqnarray}
D_{H}((\mu_1, \Sigma_1)^T, (\mu_2, \Sigma_2)^T) &=& 1 - \frac{|\Sigma_1|^{\frac{1}{4}}  |\Sigma_2|^{\frac{1}{4}}}{\left|\frac{1}{2} \Sigma_1 + \frac{1}{2} \Sigma_2\right|^{\frac{1}{2}}} \exp\left(-\frac{1}{8} (\mu_2- \mu_1)^T \left(\frac{1}{2} \Sigma_1 + \frac{1}{2} \Sigma_2\right)^{-1} (\mu_2 - \mu_1)\right)\nonumber\\
\\
D_{KL}((\mu_1, \Sigma_1)^T, (\mu_2, \Sigma_2)^T) &=& \frac{1}{2} \left(\log \frac{|\Sigma_2|}{\Sigma_1}  - d + \text{tr}\left(\Sigma_2^{-1} \Sigma_1\right) + (\mu_2 - \mu_1)^T \Sigma_2^{-1}(\mu_2 - \mu_1)\right)
\end{eqnarray}

\subsection{Estimating the mean of O-D demand}

In this sub-problem, we estimate O-D demand mean $q$, assuming the estimated variance/covariance of the link flow and O-D demand $\hat\Sigma_x,\hat\Sigma_q$ is given. The estimated variance/covariance of the link flow on those observed links $\hat\Sigma_x^o$ is also given. %The O-D mean estimation problem has different formulations under different congestion level, we will discuss two major formulations. Before that, we will discuss the properties of observed link traffic flows.

\subsubsection{Estimating probabilistic link flow on observed links}

Before estimating the O-D demand mean, we first estimate the mean of link flow on observed links, $x^o=\mathbb{E} (X^o)$, as an intermediate step from the observed data $\xx_i^o$ to the unknown demand mean $q$. If the unknown errors $\Sigma_e$ can be calibrated exogenously, the probability distribution of $X_m$ can be determined by $\N(x, \Sigma_x+\Sigma_e)$ given the probability distribution of $X$. In most cases, $\Sigma_e$ is never known, and it represents the errors that cannot be explained by the model. We then estimate $X = X_m$ with the hope that $\Sigma_e = \vec{0}$. After the estimation process, the mismatch between $\hat{X}$ and the data can be viewed as the unknown errors $E$. This will be discussed later.

The reason why we start from estimating $x^o$ is that in the classical deterministic ODE, $\norm{\Delta p q - \hat{x}^o}_2^2$ is used as the objective function, where $\hat{x}^o$ represents the daily average of traffic counts. To interpret and understand the probabilistic ODE, we need to rigorously estimate $\hat{x}^o$ before constructing an objective function. %Only when we find the distribution of $\hat{x^o}$, we can construct the MLE of O-D mean vector $q$ in a rigorous manner.

Given $\hat\Sigma_x^o$ to approximate $\Sigma_x^o$, the likelihood function of observed link flow $\xx_i^o$ can be constructed as,
\begin{eqnarray}
\label{eq:link_mle}
l(\xx^o) =\prod_{i=1}^n \frac{1}{\sqrt{(2\pi)^{|A|} |\Sigma_{x}^o|}} \exp{\left( -\frac{1}{2} (x^o - \xx_i^o)^T\left(\hat\Sigma_x^o \right)^{-1}(x^o - \xx_i^o)\right)}
\end{eqnarray}

A Maximum Likelihood Estimator (MLE) is to estimate $x^o$ to maximize the likelihood,

\begin{equation}
\label{eq:tralinkmean}
\begin{array}{rrclcl}
\vspace{5pt}
\displaystyle \min_{x^o} & \multicolumn{3}{l}{\displaystyle \frac{1}{2}  \sn  (x^o - \xx_i^o)^T \left(\hat\Sigma_x^o \right)^{-1}(x^o - \xx_i^o)} \\
\textrm{s.t.} & x^o & \geq & 0 \\
\end{array}
\end{equation}

%as in Proposition~\ref{pro:mle}.

\begin{proposition}[MLE]
\label{pro:mle}
Given the data set $\xx^o$ and link flow covariance matrix $\Sigma_x^o$ on observed links, the estimator is $\hat{x}^o = \frac{1}{n}   \xx_i^o$.
\end{proposition}

The detailed analysis of the optimization problem can be found in \ref{ap:ob}. The estimator $\hat{x}^o$ can be derived in a closed form, $\hat{x}^o = \frac{1}{n}   \xx_i^o$. The proof can also be found in \ref{ap:ob}. It is no surprise that the estimator of $x^o$ is the average of observed data $\xx_i^o$. This is consistent with the formulation of deterministic ODE that searches the best $\hat{q}$ to minimize the discrepancy between estimated link flow derived from $\hat{q}$ and daily average observed counts (namely $\frac{1}{n}   \xx_i^o$), both for those observed links.

Having the closed formulation of $\hat{x}^o$, we can further derive the probability distribution of $\hat{x}^o$.

\begin{proposition}
If the observed link flow $\xx_i^o$ i.i.d follows the probability distribution of $\N(x^o, \Sigma_x^o)$, then $\hat{x}^o$ follows:
\begin{eqnarray}
\hat{x}^o \sim \N(x^o, \frac{1}{n} \Sigma^o_x)
\end{eqnarray}
\end{proposition}

\subsubsection{Estimating demand mean on an uncongested network}
For an uncongested network, the travel costs are dependent on free-flow speed and road length, and independent of link/path flow. The route choice probability $\tilde{p}$ can be calculated exogenously. Given the route choice probability is known, we estimate the demand mean $q$ by,

\begin{equation}
\label{eq:traOD}
\begin{array}{rrclcl}
\vspace{5pt}
\displaystyle \min_{q} & \multicolumn{3}{l}{\displaystyle  n \left(\Delta^o \tilde{p} q - \hat{x}^o\right)^T \left(\hat\Sigma_x^o \right)^{-1} \left(\Delta^o \tilde{p} q - \hat{x}^o\right)}\\
\textrm{s.t.} & q & \geq & 0 \\
\end{array}
\end{equation}

Formulation \ref{eq:traOD} can be viewed as a generalized least square (GLS) formulation when regarding $n \left(\Sigma_x^o\right)^{-1}$ as the weight matrix in  \citep{robillard1975estimating, cascetta1984estimation}, which minimizes the weighted discrepancy of the mean link flow on observed links between from data $\hat{x}^o$ and from the $q$ estimator, $x = \Delta \tilde{p} q$. The above formulation can also be statistically interpreted as a maximum likelihood estimator (MLE) of the O-D demand $q$, provided with the probability distribution of the mean link flow estimator $\hat{x}^o$.

A deterministic ODE often uses Formulation~\ref{eq:simple} where the diagonal matrix $\Theta_x$ denotes the confident level for each of the observed link flow.

\begin{equation}
\label{eq:simple}
\begin{array}{rrclcl}
\vspace{5pt}
\displaystyle \min_{q} & \multicolumn{3}{l}{\displaystyle  \norm{\Theta_x(\Delta^o \tilde{p} q - \hat{x}^o)}_2^2}\\
\textrm{s.t.} & q & \geq & 0 \\
\end{array}
\end{equation}

This formulation is similar to the Formulation \ref{eq:traOD} if the row vectors of $\Delta^o$ are fully ranked, except for that the weights on each link flow observation can differ. In practice, when $\Sigma_x^o$ cannot be derived directly (e.g., due to insufficient data points), we can use the simplified Formulation \ref{eq:simple} to estimate O-D demand mean $q$.

%In most cases, the number of paths are much larger than the number of links,

Another issue for both Formulations~\ref{eq:traOD} and \ref{eq:simple} is that the optimal $\hat{q}$ may not be unique \citep{yang1994equilibrium}. To address the non-uniqueness issue, history O-D information is usually employed. An extended generalized least square can be built by assuming that history O-D information can be acquired and is independent of observed flow given $q$, then we have following formulation.

If the historical O-D demand mean and covariance matrix is given as $q^H, \Sigma_q^H$, respectively, we identify the unique solution $q$ by,
\begin{equation}
\label{eq:ODu1}
\begin{array}{rrclcl}
\vspace{5pt}
\displaystyle \min_{q} & \multicolumn{3}{l}{\displaystyle n \left(\Delta^o \tilde{p} q - \hat{x}^o\right)^T \left(\hat\Sigma_x^o \right)^{-1} \left(\Delta^o \tilde{p} q - \hat{x}^o\right) +  (q^H - q)^T \left(\Sigma_q^{H}\right)^{-1} (q^H - q)}\\
\textrm{s.t.} & q & \geq & 0 \\
\end{array}
\end{equation}
where the historical O-D covariance matrix $\Sigma_q^H$ is usually unknown. The identity matrix is used as an alternative choice.

\citet{cascetta1984estimation} proposed a similar formulation to Formulation~\ref{eq:ODu1}, but the derivation of the inverse of observed link flow covariance matrix ${\Sigma_x^o}$ is unclear. In our formulation,  ${(\Sigma_x^o})^{-1}$ can be derived analytically given $\Sigma_q$ using GESTA.

Based on Formulation~\ref{eq:ODu1}, the number of observed data $n$ and the quality of $q^H$ are two major factors affecting the accuracy of estimated O-D demand $\hat{q}$. Several remarks are made regarding $n$ and $q^H$.
\begin{remark}
Incrementing data quantity does not address the non-uniqueness issue of Formulation~\ref{eq:ODu1}.
\end{remark}
If $\Delta^o$ are fully ranked, Formulation~\ref{eq:traOD} has a unique solution for any $n \geq 1$. However, when $\Delta^o$ is not fully ranked, Formulation~\ref{eq:traOD} has multiple optimal solutions regardless the value of $n$. %So the data quantity does not help to address the non-uniqueness issue of O-D estimation
\begin{remark}
Incrementing data quantity increases the accuracy of estimated O-D demand $\hat{q}$.
\end{remark}
To simplify the discussion, we assume $\Delta^o $ is fully ranked, namely $\Delta^o \tilde{p}$ is invertible. If no observation is obtained, $n=0$, then the best estimation of $q$ would be $q^H$. This implies that the accuracy of $\hat{q}$ is solely dependent on the accuracy of $q^H$. When observations are available, the estimated O-D demand $\hat{q}$ is the weighted average of $\left(\Delta^o \tilde{p}\right)^{-1} \hat{x^o}$ and $q^H$. As $n$ increases, $\hat{q}$ becomes close to $\left(\Delta^o \tilde{p}\right)^{-1} \hat{x^o}$, otherwise to $q^H$. If we have infinite data, then $\hat{q} = \left(\Delta^o \tilde{p}\right)^{-1} \hat{x^o}$. Thus, $\Var(\hat{q}) = \vec{0}$ due to the central limit theorem (CLT) when $n=\infty$, implying $\hat{q}$ is perfectly accurate. In summary, incrementing data quantity reduces $\Var(\hat{q})$, from relying solely on the accuracy of $q^H$ ($n=0$) to being perfectly accurate ($n \to \infty$).

\begin{remark}
Formulation \ref{eq:ODu1} is not a Maximum a Posterior (MAP) estimator.
\end{remark}
Different from \citet{menon2015fine},  Formulation \ref{eq:ODu1} cannot be interpreted as a Maximum a Posterior (MAP) estimator. We first build an MLE according to the probability distribution of $\hat{x^o}$, so we need the prior of $\hat{x^o}$ to build the MAP estimaor. However, the information used in Formulation \ref{eq:ODu1} is the historical O-D $q^H$ rather than historical link flow $x^H$. Thus, formulation \ref{eq:ODu1} is not an MAP estimator. Formulation \ref{eq:ODu1} can only be interpreted as a GLS model when the history O-D is independent of observed traffic flow given $q$. Similar arguments can also be found in \citet{yang2015data}.

\begin{remark}
If $q^H \neq q$, the estimator of O-D demand mean $\hat{q}$ from Formulation \ref{eq:ODu1} is biased.
\end{remark}

To summarize, we conclude that data quantity helps increase the accuracy of the estimated O-D demand mean $\hat{q}$ while historical O-D information $q^H$ addresses the non-uniqueness issue. One subtle issue is that even if we observe a large number of data (on a large number of days), uniquely determining $\hat{q}$ may still be impossible.

\subsubsection{Estimating demand mean on a congested network}
In congested networks, the route choice probability $\tilde{p}$ is endogenously determined by GESTA. The probabilistic ODE needs to estimate the demand mean $q$ and the route choice probability $p$ simultaneously. Instead of estimating both separately, we can estimate the path flow $f  = \tilde{p}q$, analogous to the Path Flow Estimator (PFE) in the deterministic ODE settings.

Because the dimension of path flow $f$ is greater than the O-D demand $q$, using historical O-D information does not necessarily address the non-uniqueness issue in this case. In this subsection, we assume the dimension of $f$ is greater than that of link flow on observed locations $x^o$. If it is not the case, then we can always enlarge the path set by generating more paths or shrinking the observation size with network consolidation. The basic formulation is proposed in Formulation~\ref{eq:TraF},

\begin{equation}
\label{eq:TraF}
\begin{array}{rrclcl}
\vspace{5pt}
\displaystyle \min_{f} & \multicolumn{3}{l}{\displaystyle n \left(\Delta^o f - \hat{x}^o\right)^T \left(\hat\Sigma_x^o\right)^{-1} \left(\Delta^o f - \hat{x}^o\right) +   (q^H - Mf)^T \left(\Sigma_q^{H}\right)^{-1} (q^H - Mf)}\\
\textrm{s.t.} & f & \geq & 0 \\
\end{array}
\end{equation}

%Since the column vectors in path/O-D demand incidence matrix $M$ are not dependent to each other, so define $\Omega_1 = \{f | \Delta^o f = \hat{x^o}\}$ and $\Omega_2 = \{f | Mf = q^H\}$,  the solution to $f$ lies in $\Omega_1 \cap \Omega_2$ which is probably not unique, the estimation of $f$ lies in $\Omega_1 \cap \Omega_2$ which is probably not unique.

The fundamental problem resulting the non-uniqueness issue is that the number of  observed links is far smaller than the total number of paths. Even though all links are covered by surveillance, the path flow estimator $\hat{f}$ can still be non-unique. However, it is possible to restrict the feasible set of path flow $f$ by route choice models (namely equilibrium conditions), such as GESTA. The restricted formulation is presented in \ref{eq:ODwithEQ}.
\begin{equation}
\label{eq:ODwithEQ}
\begin{array}{rrclcl}
\vspace{5pt}
\displaystyle \min_{f} & \multicolumn{3}{l}{\displaystyle  n \left(\Delta^o f - \hat{x}^o\right)^T \left(\hat\Sigma_x^o\right)^{-1} \left(\Delta^o f - \hat{x}^o\right) +   (q^H - Mf)^T \left(\Sigma_q^{H}\right)^{-1} (q^H - Mf)} \\
\textrm{s.t.} & f & \in & \Phi^+ \\
\end{array}
\end{equation}
Where $\Phi^+$ is the feasible set of $f$. We adopt GESTA to model the traffic conditions and route choices since it generally works with any specific route choice models. Here we demonstrate the idea using deterministic UE and Logit/Probit-based SUE as the route choice model.
\begin{enumerate}[label=\roman*)]
\item UE-based GESTA\\
UE can be formulated as an optimization program to minimize $Z_1(f)$ \citep{sheffi1985urban},
\begin{equation}
\begin{array}{rrclcl}
\vspace{5pt}
\displaystyle Z_1(f) =  & \multicolumn{3}{l}{\displaystyle  \frac{1}{\Theta} \sum_{a} \int_{0}^{x_a} t_a(w) dw}\\
\textrm{\text{where}} & x & = & \Delta f \\
 & f & \geq & 0 \\
\end{array}
\end{equation}
Then $\Phi_{\text{UE}}^+$ is defined as:
\begin{eqnarray}
\Phi_{\text{UE}}^+ = \{f_1 | Z_1(f_1) \leq Z_1(f_2), \forall f_2 \geq 0 \text{ such that } M f_1 = M f_2\}
\end{eqnarray}
$\Phi_{\text{UE}}^+$ can also be written as a link-based or path-based variational inequality formulation \citep{smith1979existence}. Formulation \ref{eq:ODwithEQ} under UE constraints is known as Mathematical Programming with Equilibrium Constrain (MPEC) \citep{luo1996mathematical}.

\item Logit-based GESTA\\
In the Logit-based GESTA, travelers’ perception of travel costs is assumed to follow Gumbel distribution. The variance/covariance of the path cost $C$ is not considered. \citet{fisk1980some,janson1993most} cast the Logit model to its dual form, a convex optimization problem that minimizes the following objective function:
\begin{equation}
\begin{array}{rrclcl}
\vspace{5pt}
\displaystyle Z_2(f) =  & \multicolumn{3}{l}{\displaystyle  \frac{1}{\Theta} \sum_{ts} \sum_{k} f_{rs}^k \log (f_{rs}^k)}\\
\textrm{\text{where}} & f & \geq & 0 \\
\end{array}
\end{equation}
Then $\Phi^+_\text{Logit}$ is defined as:
\begin{eqnarray}
\Phi^+_\text{Logit} = \{f_1 | Z_2(f_1) \leq Z_2(f_2), \forall f_2\geq 0 \text{ such that } M f_1 = M f_2\}
\end{eqnarray}

\item Probit-based GESTA\\
In the Probit-based GESTA, travelers' perception errors follow Normal distribution \citep{daganzo1977multinomial}, as part of the probability distribution of path cost $C$. There does not exist a explicit form on $\Phi^+_\text{Probit}$. Any pair of $(p, f)$ satisfying the Probit route choice model are in $\Phi^+_\text{Probit}$. The details of Probit choice model can be found in \citet{daganzo1977multinomial, sheffi1985urban}.

\end{enumerate}

Formulation \ref{eq:ODwithEQ} is also known as the bi-level formulation of ODE with many existing studies \citep{nguyen1977estimating, fisk1984game, yang1992estimation, yang1995heuristic}. Since the solution uniqueness for the bi-level formulations varies by route choice models, we discuss them in Section \ref{sec:analysis}. Other properties of Formulation \ref{eq:ODwithEQ} are discussed in the following remarks.

\begin{remark}
Formulation \ref{eq:ODwithEQ} is non-convex.
\end{remark}
Since $\Phi^+$ is clearly not a convex set regardless of the route choice models, Formulation \ref{eq:ODwithEQ} is not convex. A sensitivity-based algorithm by \citet{josefsson2007sensitivity} and a heuristic algorithm by \citet{yang1995heuristic} are commonly used to solve for it. In addition, \citet{nie2010relaxation} relaxes the UE-based ODE to a one-level optimization problem, enhanced by \citet{shen2012new} with a convex relaxation program on a one-level optimization problem.

\begin{remark}
Logit-based SUE can be approximated by a specific Probit-based SUE.
\end{remark}
In the Logit-based SUE, the only parameter for the route choice model is the dispersion factor $\Theta$. Given $\Theta$, a Probit model with a diagonal path cost variance matrix can be used to approximate the Logit model. Details regarding the transformation can be found in \citet{greene2003econometric}.

\subsection{Estimating the variance/covariance matrix of O-D demand}
To estimate O-D demand variance and covariance matrix, we assume the link flow mean $x$ and path flow mean $f$ are provided and known. Consequently, the route choice probability $p$ and O-D demand mean $q$ are also known. We find an MLE to estimate the $\Sigma_q$. We first present the basic formulation to estimate $\Sigma_q$ given that  $\Sigma_x$ follows the Wishart distribution \citep{wishart1928generalised}. Due to the high dimension of $\Sigma_q$, Lasso \citep{tibshirani1996regression} regularized formulation is proposed to search for a sparse estimation of $\Sigma_q$ that makes trade off between variance and bias.

\subsubsection{Basic formulation}
First we define the empirical covariance matrix of the observed link flow to be $S_x^o = \frac{1}{n} \sn (\xx_i^o - \bar{\xx}^o) (\xx_i^o - \bar{\xx}^o)^T$, which is the maximum likelihood estimator of covariance matrix. $\bar{\xx}^o$ is the averaged observed link flow, $\bar{\xx}^o = \frac{1}{n} \sn \xx_i^o$. Note $S_x^o$ is different from the sample covariance matrix $P_x^o=\frac{1}{n-1} \sn (\xx_i^o - \bar{\xx}^o) (\xx_i^o - \bar{\xx}^o)^T$. Since the link flow variance-covariance matrix follows the Wishart distribution, a maximum likelihood estimator can be built to solve for $\Sigma_q$.
\begin{proposition}
Given the variance-covariate matrix for observed link flow $S_x^o$, the maximum likelihood estimator of $\Sigma_q$ is
\begin{equation}
\label{eq:TraS}
\begin{array}{rrclcl}
\vspace{5pt}
\displaystyle \max_{\Sigma_q} & \multicolumn{3}{l}{\displaystyle   \log \det ( \left(\Sigma_x^o\right)^{-1}) - \text{trace} (S_x^o \left(\Sigma_x^o\right)^{-1})} \\
\textrm{s.t.} & \Sigma_x^o & = & \Delta^o \Sigma_{f|q}\left( \Delta^o \right)^T  +\Delta^o \tilde{p} \Sigma_q \tilde{p}^T  \left(\Delta^o\right)^T \\
& \Sigma_q & \succeq & 0
\end{array}
\end{equation}
\end{proposition}

%At first look the above formulation seems to be non-convex since it involves both $\Sigma_q, \Sigma_x$ and their inverse.

The first constraint in Formulation~\ref{eq:TraS} is obtained from GESTA through Propositions~\ref{pro:gesta1} and \ref{pro:marginal}.
The convexity of Formulation~\ref{eq:TraS} depends on the rank of $\Delta^o$. If $\Delta^o$ is fully ranked, Formulation~\ref{eq:TraS} is non-convex. But if $\Delta^o$ is not fully ranked, then we can first find the optimal $\Sigma_x$ for the objective function, and then solve for  $\Sigma_x = \Delta \Sigma_{f|q} \Delta^T +\Delta \tilde{p} \Sigma_q \tilde{p}^T \Delta^T $. Both steps are convex optimization problems. In addition, $\Sigma_q$ contains $\frac{1}{2} |K_q| (|K_q|-1)$ elements, which is usually in a higher dimension than the number of observed data, so the optimal estimator of $\Sigma_q$ may not be unique.

Next we introduce the regularization and relaxation of Formulation~\ref{eq:TraS} to achieve convexity and uniqueness.

\subsubsection{Sparse model selection}
Since the number of entries in the O-D variance-covariance matrix is usually much greater than the size of data, a Lasso penalization is used to select the O-D variance-covariance matrix as in Formulation~\ref{eq:SL},
\begin{equation}
\label{eq:SL}
\begin{array}{rrclcl}
\vspace{5pt}
\displaystyle \min_{\Sigma_q} & \multicolumn{3}{l}{\displaystyle   \log \det (\Sigma_x^{o}) + \text{trace} (S_x^o  \left(\Sigma_x^o\right)^{-1}) + \lambda \norm{\Sigma_q}_1} \\
\textrm{s.t.} & \Sigma_x^o & = & \Delta^o \Sigma_{f|q}\left( \Delta^o \right)^T  +\Delta^o \tilde{p} \Sigma_q \tilde{p}^T  \left(\Delta^o\right)^T \\
& \Sigma_q & \succeq & 0
\end{array}
\end{equation}

$\lambda$ is a Lasso parameter to adjust the sparsity of $\hat{\Sigma}_q$. $\Sigma_{f|q}$ is constructed using $p$ and $\hat{q}$.  We note Formulation~\ref{eq:SL} obtains a biased but robust estimator of OD variance/covariance matrix. Formulation~\ref{eq:SL} is hard to solve due to its non-convexity \citep{bien2011sparse}. Although non-linear optimization methods can be employed to solve this formulation, none of them can guarantee computationally efficiency thus not suitable for large-scale networks. We would prefer to approximate it using a convex optimization problem with Lasso regularization. Inspired by \cite{yuan2007model}, a second order approximation to the MLE of the covariance matrix is used in Formulation~\ref{eq:LSwithLasso}.

\begin{equation}
\label{eq:LSwithLasso}
\begin{array}{rrclcl}
\vspace{5pt}
\displaystyle \min_{\Sigma_q} & \multicolumn{3}{l}{\displaystyle   \norm{S_x^o - \Sigma_x^o}_F^2 +  \lambda  \norm{\Sigma_q}_1}\\
\textrm{s.t.} & \Sigma_x^o & = & \Delta^o \Sigma_{f|q}\left( \Delta^o \right)^T  +\Delta^o \tilde{p} \Sigma_q \tilde{p}^T  \left(\Delta^o\right)^T \\
& \Sigma_q & \succeq & 0
\end{array}
\end{equation}

where $\norm{A }_F = \sqrt{\text{Tr}(A^T A)}$ and $\norm{A}_1 = \sum_{ij} |A_{ij}|$. The former one is known as Frobenius Norm, equivalent to the element-wise $L_2$ norm \citep{jennings1992matrix}. The latter one is the element-wise $L_1$ norm.
\begin{proposition}[Convexity]
The optimization problem \ref{eq:LSwithLasso} is convex.
\end{proposition}
\begin{proof}
$\Sigma_q$ can only be positive semi-definite matrix, which forms a convex set. Then we plug $\Sigma_x  =  \Delta \Sigma_{f|q} \Delta^T +\Delta \tilde{p} \Sigma_q \tilde{p}^T \Delta^T $ into the objective function. The objective function with respect to $\Sigma_q$ is also convex, so the entire formulation is convex.
\end{proof}

Formulation \ref{eq:LSwithLasso} is a desired optimization problem to solve for a sparse O-D demand variance/covariance matrix $\Sigma_q$, because the convexity allows its computational efficiency. In principle, we can use proximal methods \citep{parikh2013proximal} to solve Formulation \ref{eq:LSwithLasso}. Details of the solution algorithms can be found in \ref{ap:proximal}.

\subsection{Incorporating day-to-day travel time data in probabilistic ODE}
The day-to-day travel time/speed data can be added to the Formulation~\ref{eq:ODwithEQ} to further enhance the ODE. Real-time traffic speed data vendotrs, such as INRIX and HERE, can provide traffic speed data covering major roads in most of U.S. cities. Some studies \citep{balakrishna2006off, ma2006accelerating, kostic2015using} regarded the observed travel time/speed as another objective to minimize, and thus enhance the Formulation~\ref{eq:ODwithEQ} to become Formulation~\ref{eq:ODwithTT}.
\begin{equation}
\label{eq:ODwithTT}
\begin{array}{rrclcl}
\vspace{5pt}
\displaystyle \min_{f} & \multicolumn{3}{l}{\displaystyle  w_1 \left(\Delta^o f - \hat{x}^o\right)^T \left(\Sigma_x^o \right)^{-1} \left(\Delta^o f - \hat{x}^o\right) +  w_2 (q^H - Mf)^T \left( \Sigma_q^{H} \right)^{-1} (q^H - Mf) + w_3 (c^o - \hat{c}^o )^T  \Sigma_c^o (c^o - \hat{c}^o )  } \\
\textrm{s.t.} & f & \in & \Phi^+ \\
~ & c & = & t(\Delta f) \\
\end{array}
\end{equation}
where $t(\cdot)$ is the link performance function that maps the link flow to link costs (such as the well known BPR functions). $w_1$, $w_2$ and $w_3$ are weights assigned to each objective.

Similar to estimating the covariance matrix of link flow by Formulation \ref{eq:LSwithLasso}, one can also estimate the covariance matrix of travel cost/time given an estimator for its mean. However, a bigger issue is that the mapping from traffic speed to traffic hourly volume on the road segment is not a one-to-one mapping. A link performance function, though uniquely maps travel cost/time to volume in both ways, can be very sensitive in determining volume given near free-flow travel time. When travel speed/time data based on probe vehicles is highly biased, the error can be amplified through the link performance function. Therefore, the ODE relying on travel speed data in the static network settings is practically challenging. We believe that applying travel speed/time data can be more useful when extending GESTA and probabilistic ODE to dynamic network settings where the traffic dynamics is captured using microscopic or mesoscopic flow models. We hope to address probabilistic dynamic ODE in a future research paper. %The estimated speed depends on hyper-parameters in the link performance functions for each road segment.

%In this study we propose a reduction from Formulation~\ref{eq:ODwithEQ} to Formulation~\ref{eq:ODu1} by estimating $\tilde{p}$ exogenously from day-to-day travel time data. In this way we can bypass the equilibrium constraints by using the day-to-day travel time data. That is to say, we first estimate $\tilde{p}$ based on the observed travel time data, then apply Formulation~\ref{eq:ODu1} and Formulation~\ref{eq:LSwithLasso} to estimate the probabilistic O-D demand.

\section{Some properties of the formulations}
\label{sec:analysis}
In this section, we discuss how to evaluate the accuracy and effectiveness of our estimated O-D demand mean and variance/covariance matrix. We analyze variance/covariance by its decomposition into three main sources, O-D demand variance, route choice variance and unknown (unexplained) error. The observability of the proposed probabilistic ODE framework is also discussed.

\subsection{Goodness of fit}
For classical ODE, the goodness of fit indicator measures how close the estimated O-D demand mean $\hat{q}$ can, if loaded into the network following a deterministic traffic assignment model (UE or SUE), reproduce the observed traffic conditions. Commonly used indicators are summarized in \citet{antoniou2015towards}. Similarly, for the proposed probabilistic ODE, the goodness of fit can be measured by how close the estimated probabilistic O-D demand $\hat{Q}$, if loaded into the network following GESTA, can reproduce the probability distribution of observed traffic flow. Define the estimated link flow on observed links from the probabilistic ODE $\hat{X}^o \sim \N(\hat{x}^o, \hat{\Sigma}_x^o)$ and estimated link flow on observed links directly from data $\XX^o \sim \N(\bar{\xx}^o, P_x^o)$, $P_x^o=\frac{1}{n-1} \sn (\xx_i^o - \bar{\xx}^o) (\xx_i^o - \bar{\xx}^o)^T$. The goodness of fit indicator can be computed by the Hellinger distance or Kullback-Leibler distance between $\hat{X^o}$ and $\XX^o$.

%\subsection{Variance analysis and hypothesis test}
\subsection{Variance analysis}
In Section~\ref{sec:formulation}, we do not consider $\Sigma_e$. With the real world data, the observed data cannot be fully explained by the ODE, and thus contain unexplained errors. After the probabilistic ODE process, we can analytically decompose the link flow variance to check how much variance can be explained by the ODE.

\begin{proposition}[Link flow variance decomposition]
The variance of link flow can be decomposed into three parts, O-D demand variance, route choice variance, and unknown errors.
\begin{eqnarray}
X_m &=& x + \eta + \tau + \varepsilon_e\\
\eta &\sim& \N(0, \Delta \tilde{p} \Sigma_q \tilde{p}^T \Delta^T)\\
\tau &\sim& \N(0, \Delta \Sigma_{f|q} \Delta^T)\\
\varepsilon_e &\sim& \N(0, \Sigma_e)
\end{eqnarray}
\end{proposition}

There are many ways to quantify the variance ratio. In this study, we determine the ratio based on matrix traces, which is widely adopted in the statistics literature. Trace-based variance ratio is closely related to the spectral analysis of recurrent link/path flow data. Details and examples can be found in \citet{GESTA}.

%According to the decomposition, we can measure how much variance can be explained by our model.
%\begin{definition}
%\label{def:ratio}
%The variance that can be explained by our model is $\Delta \hat{\tilde{p}} \hat{\Sigma_q} \hat{\tilde{p}}^T \Delta^T + \Delta \hat{\Sigma_{f|q}} \Delta^T$, which is the summation of the estimated variance from O-D and estimated variance from travelers' route choice. Then the ratio of variation that can be explained by our model is define as:
%\begin{eqnarray}
%r = \frac{\text{tr}\left(\Delta \hat{\tilde{p}} \hat{\Sigma}_q \hat{\tilde{p}}^T \Delta^T + \Delta \hat{\Sigma}_{f|q} \Delta^T\right)}{\text{tr}(Q_n)}
%\end{eqnarray}
%\end{definition}

%There are many ways to quantify the ratio. In this study we determine the ratio based on matrix traces, which is widely considered to be precise and robust. Trace based variance ratio is closely related to spectral analysis of recurrent link/path flow data.

\subsection{Observability}
ODE is notoriously difficult because it is underdetermined. Studies on O-D observability problem specifically discusses the issue of solution non-uniqueness \citep{singhal2007identifiability, yang2015data}. In this subsection, we discuss the uniqueness property of the proposed probabilistic ODE.

As we discussed in Section~\ref{sec:formulation}, under no congestion, Formulation \ref{eq:ODu1}  is able to  estimate O-D mean $q$ uniquely once prior information $q^H$ is introduced \citep{bell1991estimation, yang1994equilibrium}. Studies also suggested estimate $\hat{q}$ by taking the pseudo-inverse of $\Delta^o$ that encodes a singular value decomposition (SVD) process in its formulation \citep{nie2005inferring}. For a congested network with the Formulation \ref{eq:ODwithEQ}, its solution $\hat{f}$ may not be unique. The observability of Formulation \ref{eq:ODwithEQ} varies by the constraints $\Phi^+$, dependent on the specific route choice model adopted under GESTA.

%even though we introduce the history O-D information, since $M$ shares similar property with $\Delta^o$, as we discussed above that if $\Omega_1 = \{f | \Delta^o f = \hat{x^o}\}$ and $\Omega_2 = \{f | Mf = q^H\}$, the solution to $\hat{f}$ lies in $\Omega_1 \cap \Omega_2$ which is probably not unique.

\begin{enumerate}[label=\roman*)]
\item UE-based GESTA\\
Generally path flow under the UE condition is not unique given O-D demand $q$ \citep{smith1979existence}. When UE is used as the constraint for Formulation \ref{eq:ODwithEQ}, it cannot guarantee the optimal path flow estimator $\hat{f}$ to be unique. However, we can find an extreme point solution from the feasible domain, and then this solution to Formulation \ref{eq:ODwithEQ} is unique \citep{tobin1988sensitivity}. The extreme point solution can be obtained through column generation. The optimal estimator $\hat{f}$ is also locally stable and the upper level object function is strongly convex \citep{tobin1988sensitivity, patriksson2004sensitivity}. Thus, if we use history O-D demand as the initial point and the history O-D is near the true O-D demand, then the optimization process is likely to find the optimal solution without trapping into a local minimum \citep{yang1995heuristic}.
\item Logit-based GESTA\\
Since the Logit model is strictly convex on $f$, the optimal solution to Formulation \ref{eq:ODwithEQ} is unique \citep{patriksson2004sensitivity}.
\item Probit-based GESTA\\
Since $\hat\Sigma_x$ is given in Formulation \ref{eq:ODwithEQ}, the probability distribution of path costs $C$ is uniquely determined. As a result, the solution to the route choice probability $p$ is unique under the Probit model. If the history O-D information is used, the optimal estimator of O-D demand mean $\hat{q}$ is unique. Since both $p$ and $\hat{q}$ are unique, the optimal estimator of path flow $\hat{f}$ is also unique.
\end{enumerate}

As for Formulation \ref{eq:LSwithLasso} to estimate O-D demand variance/covariance matrix, the optimal solution is non-unique since Lasso regularization is not strictly convex \citep{tibshirani2013lasso}. However, practically, Lasso regularization can largely shrink the solution domain towards being unique.

\begin{proposition}
\label{pro:nonunique}
The optimal solution $(\hat{q}, \hat{\Sigma}_q, \hat{f}, \hat{x})$ to the IGLS framework consisting of both formulations \ref{eq:ODwithEQ}, \ref{eq:LSwithLasso} may not be unique for any given observed link flow data set $\xx^o$.
\end{proposition}

The major reason for the non-uniqueness is that the number of rows (namely the number of links that are covered with sensors) in $\Delta$ or $\Delta^o$ is much smaller than the number of its columns (namely the number of paths) in a general large-scale network. Consequently, $\Delta^o P q$ is unique, but $q$ is not unique in Formulation \ref{eq:traOD}. $\Delta^o f$ is unique, but $f$ is not unique in Formulation \ref{eq:TraF}. Similarly, $\Sigma_x$ is unique, but $\Sigma_q$ is not unique in Formulation \ref{eq:TraS}. Generally, the entire IGLS framework estimates both the mean and variance/covariance matrix, and thus has to search a much larger domain space than a deterministic ODE. Therefore, its observability is worse off. This will be further demonstrated in the numerical experiments.

Though Proposition~\ref{pro:nonunique} declares a challenge for estimating probabilistic O-D demand, we argue that by the proposed IGLS framework, the O-D mean estimator is no worse than a best possible estimator by an error that reduces with respect to the sample size, and thus no worse than the O-D demand estimator using deterministic ODE methods. %The formal statement is Proposition~\ref{pro:better}.
\begin{proposition}
\label{pro:better}
Suppose observations of link flow on observed links that are i.i.d drawn from the probability distribution of $X$ on each day, and they are used to estimate the mean and variance/covariance matrix of link flow. $\hat{\Sigma}_x \succ 0$. For an arbitrary route choice probability vector $p\geq 0$ (or equivalently an underlying route choice model), the statistical risk of the estimated O-D mean $\hat q$ from Formulation~\ref{eq:traOD} (or Formulation~\ref{eq:ODwithEQ} without history O-D information) is of ${\cal O}\left(\frac{1}{n}\right)$ where $n$ is the sample size (namely the number of days with observations).
\end{proposition}
\begin{proof}[Proof sketch]
Note that the observed link set $A^o$ does not change from day to day. First we define a risk function to measure the performance of a specific estimator. The risk is low when the estimator provides an accurate solution given any observed link flow data from a probability distribution of $X$, whereas the risk is high when the estimator is either inaccurate or not robust to the observed link flow. We then rewrite the ODE formulation and bound the risk. We show that the risk of the estimator for O-D demand mean is of ${\cal O}(\frac{1}{n})$ regardless of the quality of the estimated link variance/covariance matrix $\Sigma_x$. A detailed proof is provided in \ref{ap:risk}.
\end{proof}

Proposition~\ref{pro:better} is one of the major features that distinguish this research from other existing probabilistic ODE methods. Though probabilistic ODE works with a much larger solution space than the deterministic O-D ODE, Proposition~\ref{pro:better} guarantees that using the proposed IGLS framework, the estimator for O-D demand mean is no worse than the best possible estimator by ${\cal O}(\frac{1}{n})$, and thus the mean estimated by deterministic ODE. Provided with a large data sample, the proposed probabilistic ODE will not ``get lost" due to enlarged searching solution domain regardless of the variance/covariance matrix. In other words, the estimator for O-D variance/covariance matrix can be seen as additional information to be inferred using day-to-day traffic data, in addition to the mean estimator. The variance/covariance matrix does not impair the performance of estimated O-D mean vector. This is one critical feature that distinguishes our research from \citet{shao2014estimation}, where the formulation solving for both mean and covariance matrix simultaneously may not necessarily guarantee a robust estimator for the O-D demand mean.

%In summary, we proved that probabilistic O-D estimation problem is harder than deterministic O-D estimation problem, but the proposed IGLS can guarantee the estimated O-D mean is no worse than that estimated from deterministic O-D estimation methods.

\section{Solution algorithms}
\label{sec:soluion}
In this section we present the solution algorithm for the proposed IGLS framework. The goal is to compute the estimators for O-D mean and variance/covariance matrix $(\hat{q}, \hat{\Sigma}_q)$. The proposed formulations are path based. The number of paths with positive flow increases exponentially when the network grows. For small networks, path enumeration is possible. When the networks are large, we can simply enumerate $K$ shortest paths \citep{yen1971finding, eppstein1998finding} for each O-D pair and then search for the solution in the prescribed path set. In addition, the proposed IGLS framework can also fit the column generation method \citep{watling2015stochastic, rasmussen2015stochastic}. At each iteration, one or several additional paths that possess minimal path cost at the time of iteration can be generated and added to the prescribed path set.

For the sub-problem of estimating O-D demand mean vector $q$, two heuristic algorithms can be used to directly solve the bi-level formulation \citep{yang1995heuristic, josefsson2007sensitivity}. A single-level convex relaxation to the formulation can also be adopted \citep{shen2012new}. In addition, two algorithms, Iterative Shrinkage-Thresholding Algorithm (ISTA) \citep{nesterov1983method} and Fast Iterative Shrinkage-Thresholding Algorithm (FISTA) \citep{nesterov2005smooth} solve for the sub-problem of estimating O-D demand variance-covariance matrix $\Sigma_q$.

The solution algorithm is summarized as follows,

\ \\
%\begin{table}[h]
  \begin{tabular}{p{4cm}p{10cm}}
  \textbf{Algorithm}&  \\\hline
  \textit{Step 0} & \textit{Initialization.} Iteration $\nu=1$, generate a path set for each O-D pair. Set the initial value of estimated O-D mean and variance/covariance matrix $(\hat{q}, \hat{\Sigma}_q)$. \\\hline
  \textit{Step 1} & \textit{Estimating O-D demand mean.}  Fix $\hat{\Sigma}_q$ and $\hat{\Sigma}_x$, and estimate path flow and O-D demand following Formulation~\ref{eq:traOD} or \ref{eq:ODwithEQ}.\\\hline
  \textit{Step 2} & \textit{Estimating O-D variance/covariance matrix.}  Fix the estimated path flow $\hat{f}$, estimate O-D variance/covariance matrix following Formulation~\ref{eq:LSwithLasso} and \ref{ap:proximal}.\\ \hline
  \textit{Step 3} & \textit{Network loading.}  Perform the network loading according to the current assignment $p$ to obtain flow $\hat{x}, \hat{\Sigma}_x, \hat{f}, \hat{\Sigma}_f$.\\\hline
  \textit{Step 4} & \textit{Convergence check.} Check the estimated probabilistic O-D $(\hat{q}, \hat{\Sigma}_q)$. If the convergence criterion is met, go to \textit{Step 6}; if not, $\nu=\nu+1$, go to \textit{Step 1}.\\ \hline
  \textit{Step 6} & \textit{Output.} Output $(\hat{q}, \hat{\Sigma}_q)$. \\ \hline
  \end{tabular}
%  \caption{Alternating Optimization STA Method}
%\end{table}
\ \\

\section{Numerical experiments}
\label{sec:exp}
We first examine the probabilistic ODE method on two small networks. Results are presented, discussed, and compared among different route choice models. The sensitivity of the data quantity and historical O-D information are also tested and analysed. The impact of penalty term on Lasso regulrazation is analysed. We also compare the efficiency of different gradient-based methods to solve Formulation~\ref{eq:LSwithLasso}. In addition, the proposed method is also applied on two large-scale real-world networks to examine its efficiency and scalability.

The true O-D demand in the real world is notoriously difficult to obtain. We adopt two methods to validate the proposed probabilistic ODE. In the first method, we construct probabilistic O-D demands then regard it as the ``true'' O-D demand. We run GESTA with Probit as the route choice model to obtain the probability vector $p$, again as the ``true'' choice probabilities. Then we randomly sample a set of O-D demand from its probability distribution, further randomly sample the route choice for each of the trips to obtain link/path flow, and finally add a perturbation of error (as the unknown error) to the link flow. The perturbation of error from $-\varepsilon$ to $\varepsilon$ is  generated as follows: the perturbed value is $\xi_p = \xi (1 + \texttt{rand}) \varepsilon$ when the actual value is $\xi$, where $\texttt{rand}$ is a sample uniformly distributed between $[-1, 1]$. We do this random sampling for a sequence of many trials, each of which is seen as the observation on one day. A subset of the link flow is used as the observations. The performance of the probabilistic ODE is assessed by comparing the estimated O-D demand to the ``true'' O-D demand \citep{antoniou2015towards}. In the second method, we estimate the probabilistic O-D demand using real-world day-to-day traffic flow count data using this IGLS framework, and check if the estimated O-D demand, along with GESTA, can reproduce the actual traffic flow observed for a set of days.

To measure the error of estimated O-D mean, we use Percentage Root Mean Square Error (PRMSE). Kullback-Leibler distance (KL distance) is used to measure the error of estimated O-D probability distribution.
\begin{eqnarray}
PRMSE(q, q^{true}) &=& 100 \% \times \sqrt{\frac{\sum_{rs}\in K_q (q_{rs} - q_{rs}^{true})^2}{|K_q|}} \times \frac{|K_q|}{\sum_{rs \in K_q} q_{rs}^{true}}\\
D_{KL}((q, \Sigma_q)^T, (q^{true}, \Sigma_q^{true})^T) &=& \frac{1}{2} \left(\log \frac{|\Sigma_q^{true}|}{|\Sigma_q|}  - d + \text{tr}\left((\Sigma_q^{true})^{-1} \Sigma_q\right) + (q^{true} - q)^T (\Sigma_q^{true})^{-1}(q^{true} - q)\right)\nonumber\\
\end{eqnarray}

In the numerical experiments, we test a few different settings for the sub-problem of estimating O-D demand mean, as well as for the sub-problem of estimating variance/covariance. ``w/o EC" implies that Formulation~\ref{eq:TraF} is adopted without an equilibrium constraint. We use ``Logit" and ``Probit'' to denote Logit-based GESTA and  Probit-based GESTA, respectively. When estimating O-D demand variance, ``w/o Lasso" represents the setting without the Lasso regularization and ``w/ Lasso" with the Lasso regularization.

\subsection{A small three-link network}
We first work with a toy network with three links, three paths and two O-D pairs as shown in Figure~\ref{fig:net2}. The Bureau of Public Roads (BPR) link travel time function is adopted,
\begin{eqnarray}
t_a(X_a) = t^0_a \left[1 + \alpha \left(\frac{X_a}{cap_a}\right)^\beta \right]
\end{eqnarray}
where $t^0_a$ is the free-flow travel time on link $a \in A$, $\beta = 4$, $\alpha = 0.15$ are constant parameters. $cap_a$ denotes the capacity of link $a$. Link settings are $t^0_1 =  t^0_2 = 10, t^0_3 = 5$, $cap_a = 360,\forall a\in A$. OD pairs $(1\to3)$ and $(2\to3)$ are considered, demand means are $q_{1\to3} = 700, q_{2\to3} = 500$. The variance of the O-D demand is set to be $25\%$ of the O-D demand mean. Probit-based GESTA is used as the ``true'' underlying statistical traffic assignment model.

\begin{figure}[h]
\centering
\includegraphics[scale = 0.7]{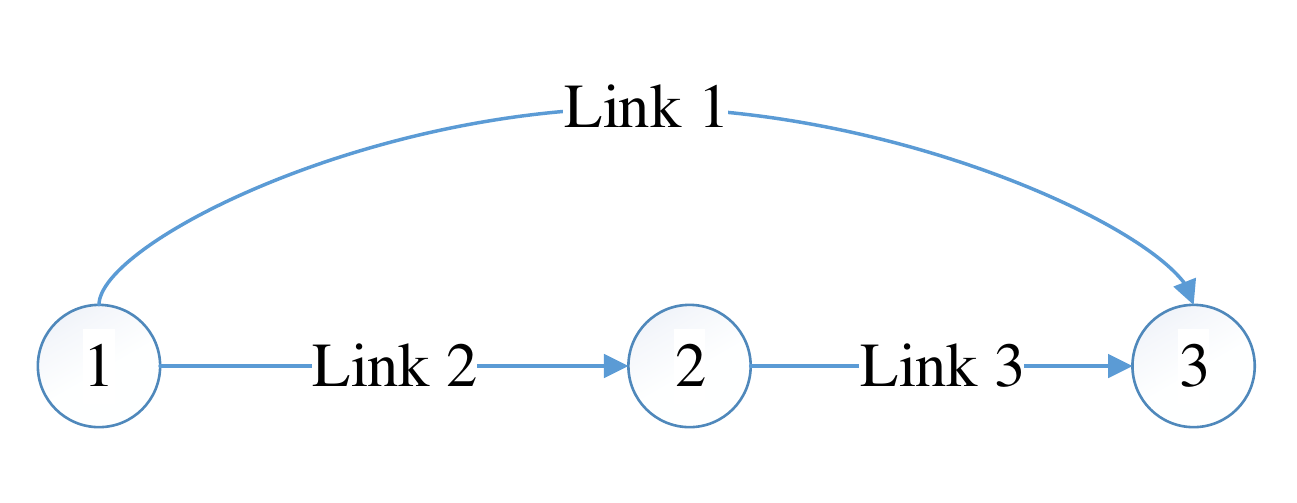}
\caption{\footnotesize A three-link toy network}
\label{fig:net2}
\end{figure}

\begin{figure}[h]
\centering
\includegraphics[scale = 0.6]{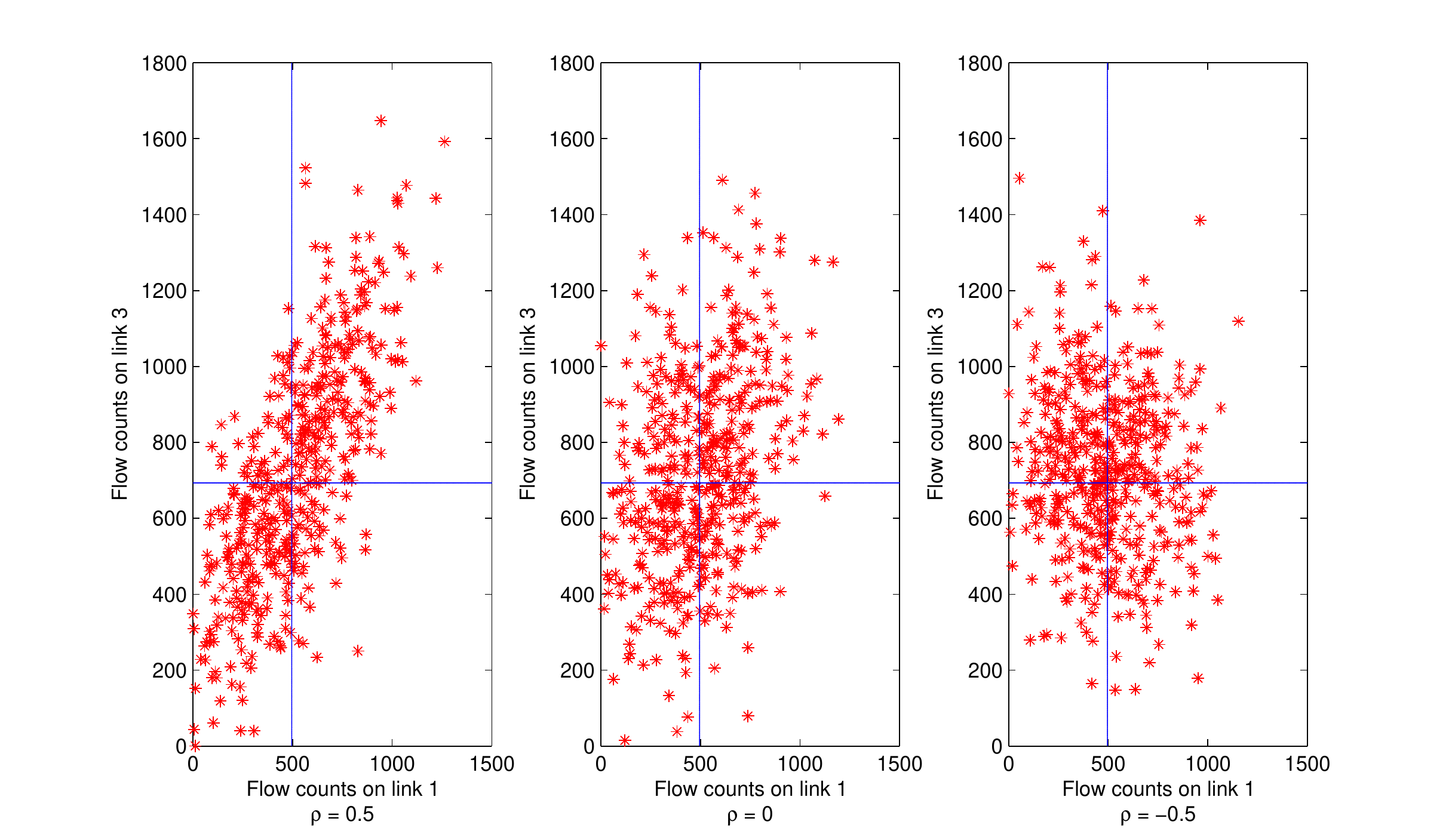}
\caption{\footnotesize Synthesized ``true'' link flow data for different correlation $\rho$}
\label{fig:2ddata}
\end{figure}

\subsubsection{Estimation results}
\label{sec:small_result}
Suppose we observe $500$ days' traffic counts on Link $1$ and Link $3$. Those observations are used to estimate the probabilistic O-D demand by different modeling settings. We tested three demand patterns where the ``true'' synthesized correlation of the demand between the two O-D pairs $\rho$ is $-0.5, 0, 0.5$ respectively. For each correlation, the observed data is synthesized and presented in Figure~\ref{fig:2ddata}. The ODE would not know the true underlying Probit route choice model. Instead, ODE would ``speculate'' a specific route choice model, Logit, Probit or no route choice model, to examine their respective performance. Furthermore, the historical O-D information is unknown. Given traffic observations of 500 days and a speculated route choice model, we apply the proposed probabilistic ODE method to estimate the probability distribution of O-D demand, all results are presented in Table~\ref{tab:basic_exp}. Note that when we apply the ODE with the Logit model, we identify the dispersion factor $\Theta=1$ such that it produces the best estimation performance of all possible values for $\Theta$. In fact, the value of $\Theta$ only marginally impacts the ODE performance, so the choice of its value does not affect our findings as much.

\begin{table}[h!]
\begin{center}
    \caption{\footnotesize Results of probabilistic ODE on the four-link toy network (no historic O-D demand information is used)}
    \label{tab:basic_exp}
    \begin{tabular}{c|cccccccc}
    \hline
    True $\rho$ & Settings &  $\hat{q}_{1 \to 3}$ & $\hat{q}_{2 \to 3}$ & $\hat{\sigma}^2_{1\to 3}$ & $\hat{\sigma}^2_{2\to 3}$ & $\hat{\rho}$ & RMPSE & KL-distance \\
    \hline \hline
    ~ & True value &  $700$ & $500$ & $175$ & $125$ & NA & NA & NA \\
    \hline
    \multirow{3}{*}{$0.5$} & w/o EC - w/o Lasso & $722.17$ & $500.41$ & $186.69$ & $134.21$ & $0.56$ & $3.62 \%$ & $3.64$ \\
                           & Logit - w/o Lasso     & $682.36$ & $499.63$ & $207.94$ & $134.21$ & $0.50$ & $2.08 \%$ & $1.17$ \\
                           & Probit - w/o Lasso     &  $699.50$   & $499.63$   & $200.94$    & $134.21$    & $0.52$  & $0.07 \%$     & $0.01$           \\
    \hline
    \multirow{5}{*}{$0$} & w/o EC - w/o Lasso & $715.91$ & $500.46$ & $143.05$ & $ 138.74$ & $0.03$ & $1.87 \%$ & $0.74$ \\
                           & Logit - w/o Lasso     & $681.28$ & $500.46$ & $162.49$ & $ 138.75$ & $0.02$ & $2.21 \%$ & $1.01$ \\
                           & Probit - w/o Lasso     & $700.30$    & $ 500.46$    & $152.15$    & $ 138.75$    & $0.03$   & $0.06 \%$     & $0.01$           \\
                           & Logit - w/ Lasso     & $681.28$ & $500.46$ & $ 144.52$ & $128.75$ & $0.00$ & $2.21 \%$ & $1.01$ \\
                           & Probit - w/ Lasso     & $700.02$ & $ 500.46$ & $ 132.27$ & $ 128.75$ & $0.00$ & $0.05 \%$ & $0.004$ \\
    \hline
    \multirow{3}{*}{$-0.5$} & w/o EC - w/o Lasso & $703.41$ & $499.06$ & $173.34$ & $132.60$ & $-0.41$ & $0.43 \%$ & $0.04$ \\
                           & Logit - w/o Lasso     & $681.05$ & $499.06$ & $184.13$ & $132.60$ & $-0.39$ & $2.23 \%$ & $1.47$ \\
                           & Probit - w/o Lasso     & $701.71$ & $499.06$ & $ 174.19$ & $ 132.60$ & $-0.41$ & $0.23 \%$ & $0.02$ \\
    \hline
    \end{tabular}
\end{center}
\end{table}

First of all, we obtain not only the O-D mean vector but also the O-D variance/covariance (namely the correlation in this example) by applying the probabilistic ODE to make the best use of all 500 days' data. As can be seen in Figure~\ref{fig:2ddata}, the daily average flow counts are the same for different $\rho$. If we apply the deterministic ODE, then the 500 days' data are taken the average to estimate the O-D demand mean. Day-to-day traffic data are not fully used, and the correlations of O-D demand among O-D pairs are overlooked. The probabilistic ODE provides more insights of O-D demand that would be needed for both transportation planning and operation.

When comparing different settings of the probabilistic ODE, all those settings yield acceptable estimations. The RMPSEs of the probabilistic ODE, under different true demand patterns, route choice models, and whether equilibrium constraints are considered in PFE, are all below $4\%$. In general, the probabilistic ODE with different settings accurately estimates not only the mean but also the variance and covariance of O-D demand.

Since the ``true'' link flow is generated based on the Probit model, it is no surprise that the ODE with the Probit model yields the best performance among all route choice models, consistently for all $\rho$ values. The ODE with the Logit model also provides a good estimate, close enough to the ODE with the Probit model, as long as the dispersion factor $\Theta$ is set properly. However, the Logit based ODE is biased since it only approximates, but cannot fully capture the true demand correlation between the two O-D pairs. If no equilibrium constraints are used, the accuracy of estimate O-D demand can go both ways. This is also no surprise as a result of a much larger domain space for the path flow compared to Logit-based or Probit-based GESTA. For this small network, it outperforms the Logit model when $\rho = -0.5$ but is less accurate when $\rho = 0.5$.

%As can be seen, Probit based method has better estimation on O-D demand while Logit based method has better estimation on O-D variance/covariance matrix. The data is generated by Probit model, so using itself will achieve better estimation of O-D. But since the route choice probability estimated from Logit model is more stable, hence make the estimation of variation easier.

Lasso regularization leads to a more accurate estimation when the true demand is independent between the two O-D pairs. The sparse O-D variance/covariance can be better interpreted with more insights, and may furthermore allow causal inference/analysis of trips made among traffic analysis zones. Lasso regularization will, however, make the estimator biased, a disadvantage of using sparse variance/covariance matrix. This is why the ODE performance with Lasso when the correlation is 0.5 or -0.5 is worse than the ODE without Lasso. If all the variance and covariance are substantial, using Lasso leads to a substantial bias for the estimator. Note that variance and bias is always a trade-off to make in the probabilistic ODE. When the weight parameter $\lambda$ for Lasso is carefully chosen (for example by cross validation), the bias may be small in exchange for a much more reliable estimator (namely a much smaller variance) comparing to the settings without Lasso regularization.

\subsubsection{Variance decomposition}

After obtaining the estimated probabilistic O-D demand, we can conduct variance analysis for the flow on each link. As an example to demonstrate the variance decomposition, we consider the estimated probabilistic O-D demand using observations drawn from $\rho = 0.5$ and the Probit model as the route choice model (listed in Table~\ref{tab:basic_exp}). The decomposition of link flow variance is presented in Figure~\ref{fig:decp}. Most of the day-to-day flow variance on Link $2$ comes from the route choice variation. In other words, change in demand levels does not affect the link flow on Link 2 as much. This is because Link $1$ has lower cost than Links 2 and 3 combined. Demand from node $1$ will prefer using Link $1$. As a result, the flow variance ratio attributed to O-D demand on Link $2$, $p_2^2 q_{1 \to 3}$, is low (where $p_2$ is the probability of choosing Link $2$ for demand from 1 to 3), while the flow variance ratio attributed to route choices, $p_1 p_2 q_{1 \to 3}$, is much higher (where $p_1$ is the probability of choosing Link $1$ for demand from 1 to 3). Link 1 is not used by the demand from node 2 to 3, and thus this demand and its route choice have an indirect impact on the flow variation on Link 1. Comparing to Link 1, Link 3 is directly affected by the demand of both O-D pairs. Consequently, Link $3$ has the highest day-to-day flow variance of all links.

\begin{figure}[h]
\centering
\includegraphics[scale = 0.35]{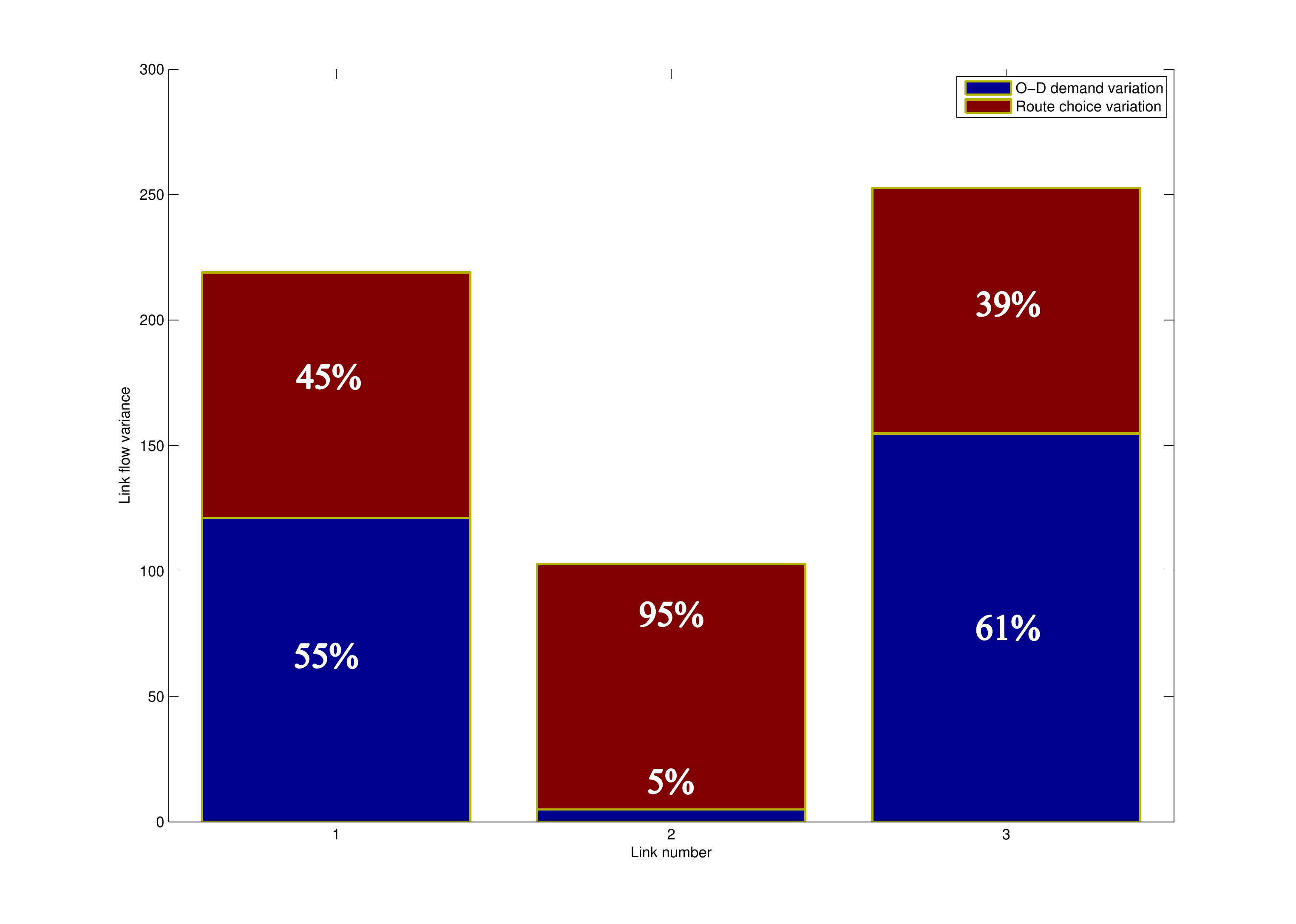}
\caption{\footnotesize Link flow variance decomposition}
\label{fig:decp}
\end{figure}

\subsubsection{Sensitivity analysis}
So far, we have not used any historical O-D demand information in the probabilistic ODE. Historic O-D information may have substantial impact to the performance of the probabilistic ODE. We now examine how sensitive the ODE results are with respect to the quality of historical O-D demand mean $q^H$. We change the historic O-D demand from $70 \%$ to $130\%$ of the ``true'' synthetic O-D demand. Other settings are the same as in previous experiments. For three ODE settings (no route choice, Probit and Logit), the KL distance from the estimated O-D demand to the ``true'' O-D demand are shown in Figure~\ref{fig:history} with respect to the quality of historical O-D demand data, represented by how far the historical O-D demand mean provided to the ODE is away from the ``true'' O-D demand mean. The variance/covariance matrix of the historical O-D demand is set to the identity matrix.

Shown in Figure~\ref{fig:history}, the probabilistic ODE without equilibrium constraints is the most sensitive to biased history O-D demand mean, as a result of a larger domain space. Its resultant KL distance is almost twice as much as the ODE with Probit-based GESTA, given the same inaccurate historical O-D demand mean. Provided with the accurate historical O-D demand mean, the ODE of both with Probit-based GESTA and without equilibrium constraints can accurately estimate the O-D demand mean. However, this is not the case for ODE with Logit-based GESTA. Generally, the ODE with Logit-based GESTA results a less KL distance than the ODE with Probit-based GESTA when the provided historical O-D demand is less than the ``true'' demand, and it results a great KL distance when the provided historical O-D demand is greater than the ``true'' demand. Clearly, the ODE with Logit-based GESTA embeds a prior bias on the demand mean, and tends to estimate more demand than the provided historical demand.  Overall, the ODE with Probit-based GESTA looks like the most robust estimator to the historical information in this case, possibly because it happens to use the ``true'' route choice model, namely the Probit model.
% and towards number of data is presented in  Figure~\ref{fig:data}.

\begin{figure}[h]
\centering
\includegraphics[scale = 0.5]{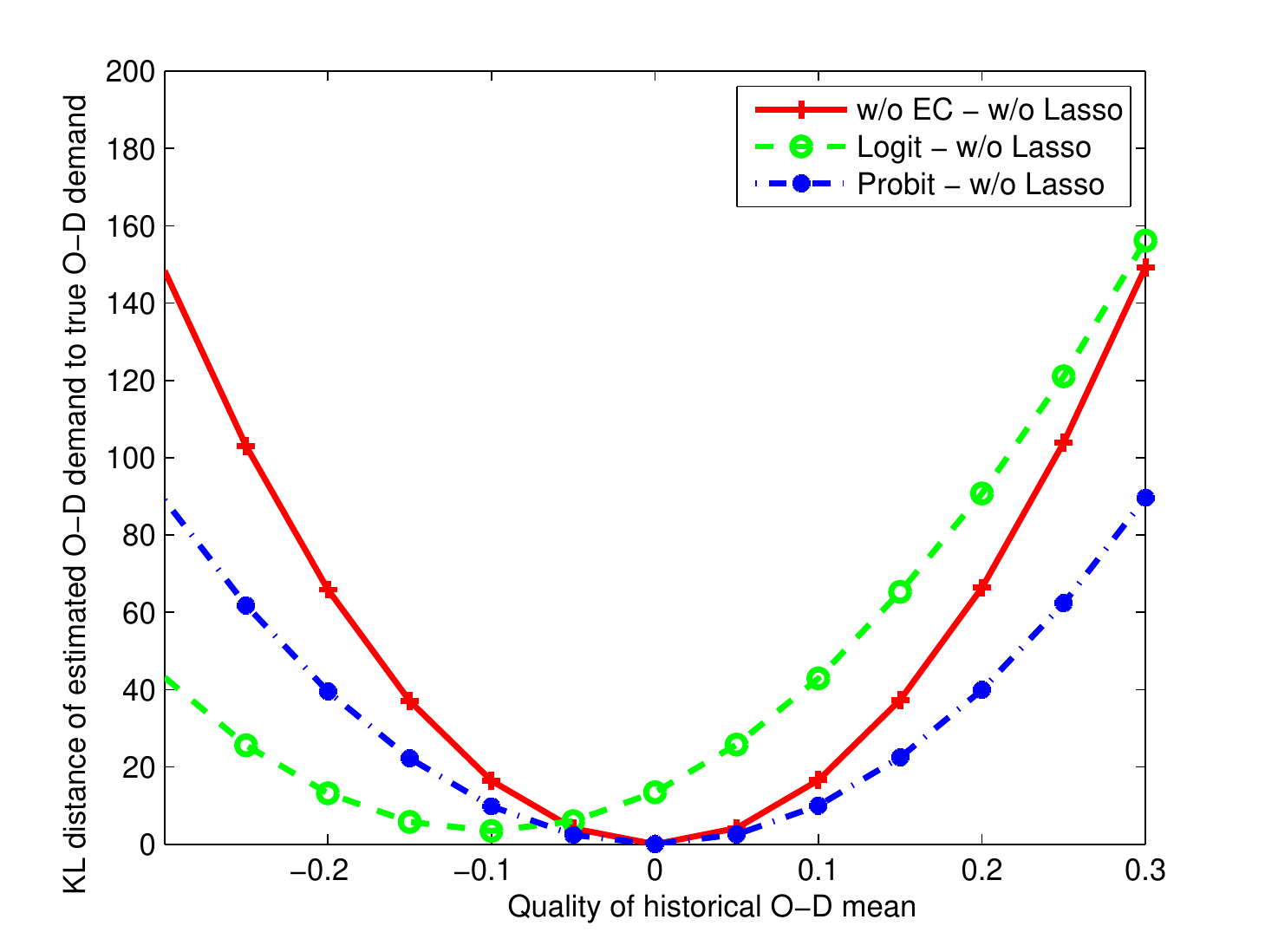}
\caption{\footnotesize Sensitivity analysis on the quality of history O-D demand mean}
\label{fig:history}
\end{figure}

Next, we examine the impact of data sample size to the ODE results. In this experiment, we again do not use the historical O-D information. The ODE without equilibrium constraints does not seem to improve as the sample size increases. This implies that 100 days of observation, in this case, does not necessarily guarantee reliable O-D demand estimate. However, if the GESTA with a route choice model is adopted, then increasing sample size can improve the ODE results. In this case, when the data sample size is more than 20, the ODE results are reasonably good. And when the size reaches 70, the ODE can provide the solution very close to the ``true'' probability distribution of the O-D demand.

\begin{figure}[h]
\centering
    \begin{subfigure}[b]{0.475\textwidth}
        \includegraphics[width=\textwidth]{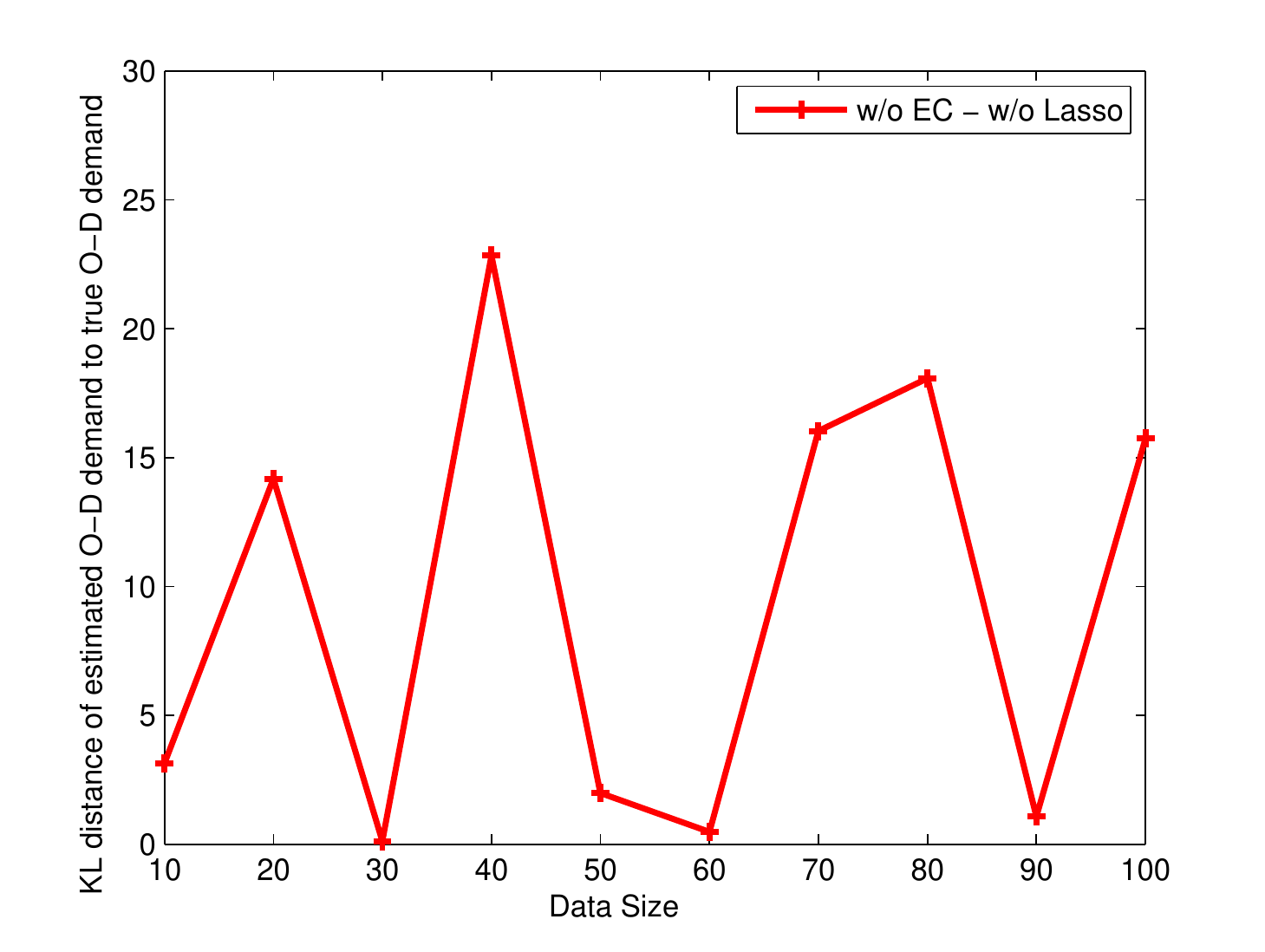}
    \end{subfigure}
    \begin{subfigure}[b]{0.475\textwidth}
        \includegraphics[width=\textwidth]{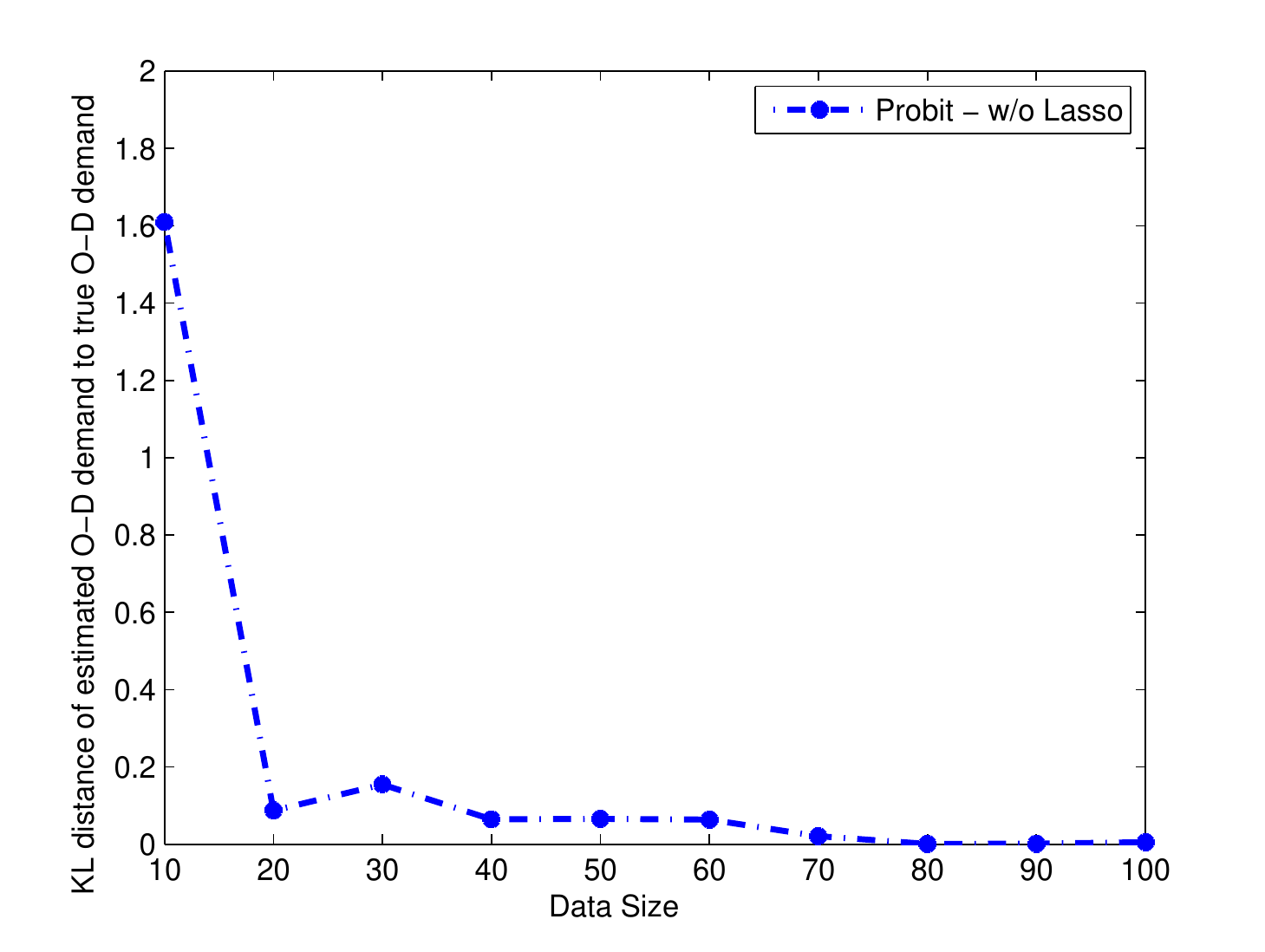}
    \end{subfigure}
    \caption{\footnotesize Sensitivity analysis on the data sample size}
    \label{fig:data}
\end{figure}

%\subsubsection{Stability and Uniqueness}
%
%To demonstrate the stability and uniqueness of the proposed probabilistic ODE, we run $100$ ODE trials, each of which is input with: 1) different initialized O-D demand mean and variance/covariance matrix; 2) re-synthesized $500$ data points. The method adopted is the Probit-based GESTA with neither history O-D information nor Lasso regularization used. The estimated results are presented in Figure~\ref{fig:random}. As we discussed in Section~\ref{sec:analysis}, the solution should be unique since we have two observations and two O-D demands to estimated, thus the 100 estimations results are located at almost the same position. So the initialization of O-D demand does not effect the final estimation results. However, changing the data set does effect the estimation results as can be seen in the right plot. All the points lines on one line since observation of link $1$ and link $3$ constrained the total demand level, thus the estimated O-D demand vary along the line which specify the total O-D demand to be $1200$.
%
%\begin{figure}[h]
%	\centering
%	\includegraphics[scale = 0.5]{random}
%	\caption{$100$ estimated O-D mean under different random initialization}
%	\label{fig:random}
%\end{figure}

\subsubsection{Efficiency of solution algorithms: ISTA versus FISTA}

We now examine the performance of solution algorithms, ISTA and FISTA,  when solving Formulation~\ref{eq:LSwithLasso}. Changes in the objective function value against the number of iterations using both algorithms are presented in Figure~\ref{fig:fista}. Clearly, FISTA is more efficient in solving the minimization problem than ISTA. In the following experiments, we use FISTA as the sole solution algorithm.

\begin{figure}[h]
\centering
\includegraphics[scale = 0.6]{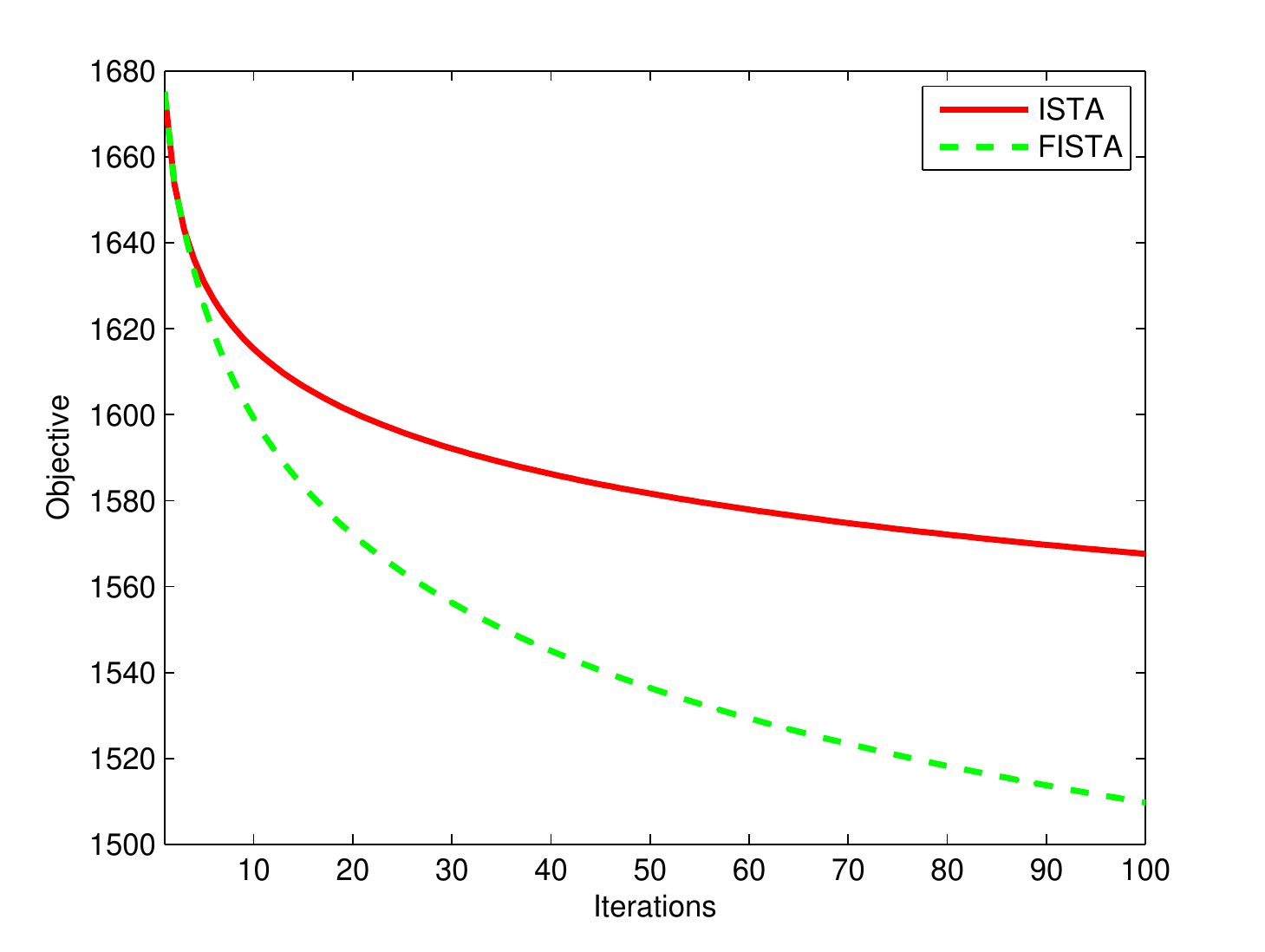}
\caption{\footnotesize Efficiency of solution algorithms: ISTA versus FISTA}
\label{fig:fista}
\end{figure}

\subsection{A second small network}
A second small toy network is used to demonstrate the effects of Lasso penalty term $\lambda$ on the estimation results. This toy network contains $6$ links, $7$ nodes and $5$ O-D pairs, as shown in Figure~\ref{fig:network3}. All 5 O-D pairs share the same link, Link 1. The free-flow travel time $t_a^0$ is $10$ for each link $a$. The capacity for Link $1$ is 1,800, and the capacity of Links 2, 3, 4, 5 and 6 are randomly drawn from $250$ to $750$. The standard BPR link performance function is adopted. The ``true'' O-D demand is synthesized by randomly drawing from $300$ to $700$ for each O-D pair. The ``true'' O-D variance/covariance matrix is set to contain $6$ zero entries out of the total $25$ entries. We synthesize the ``true'' variance of the O-D demand using its mean, and the ``true'' correlation factor is randomly drawn from $-0.5$ to $0.5$.

\begin{figure}[h!]
	\centering
	\includegraphics[scale = 0.3]{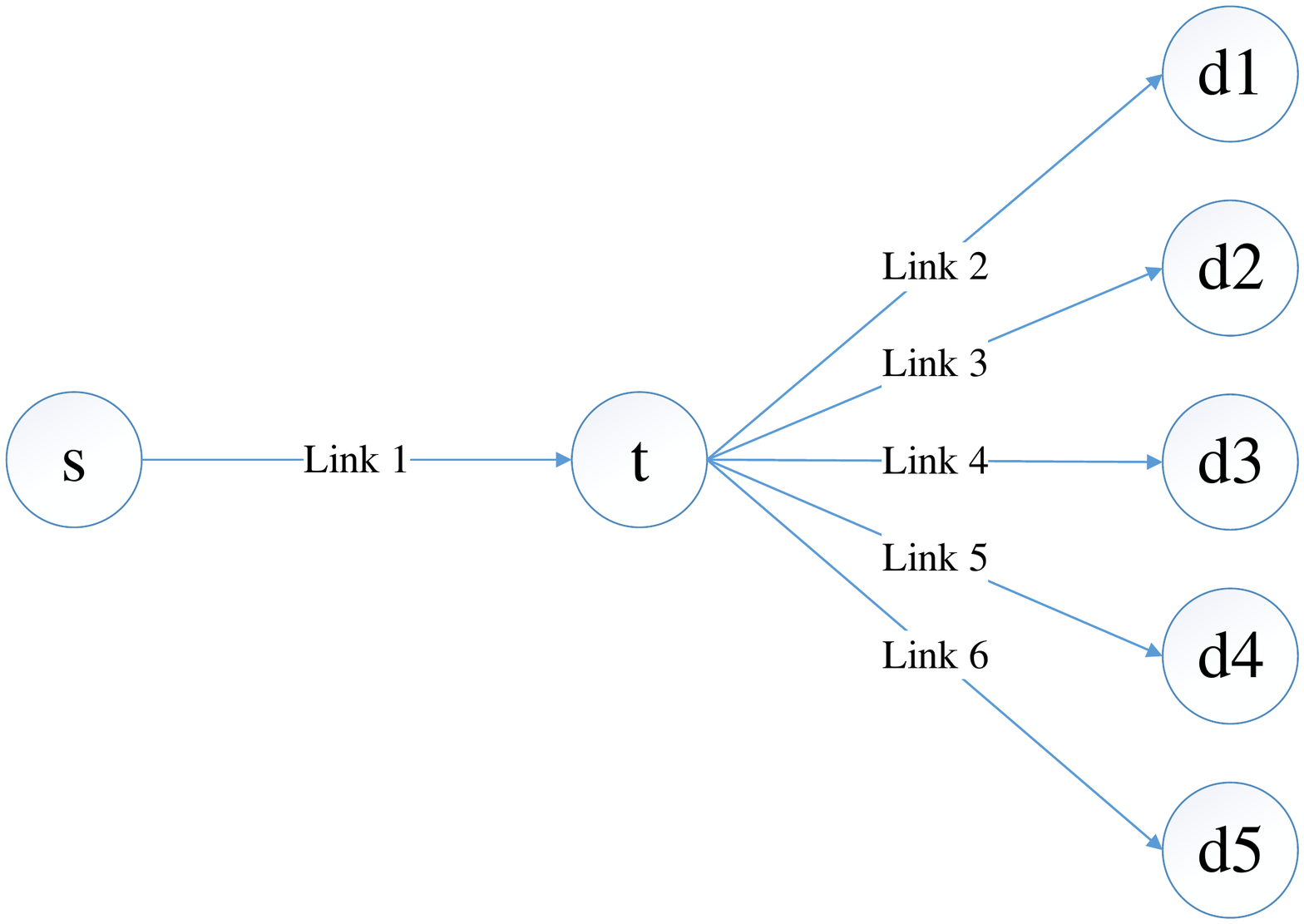}
	\caption{\footnotesize A second small network with five O-D pairs}
	\label{fig:network3}
\end{figure}

Suppose we have a full coverage on the network, namely, the flow counts on all the links are observed for in all $1000$ days. We estimate the probabilistic O-D demand using the proposed IGLS framework without using historical O-D information. Since we have a full coverage for all links, the estimated mean of the O-D demand is fairly accurate, whereas the estimated variance/covariance matrix is largely dependent on the LASSO penalization term $\lambda$. The relation between values of $15$ variance/covariance entries (due to the symmetry of this matrix) and $\lambda$ is presented in Figure~\ref{fig:path}, also known as coefficient paths for LASSO. Each path represents the value of one entry in the O-D variance/covariance matrix under different LASSO penalty $\lambda$. As can be seen from Figure~\ref{fig:path}, when the Lasso penalty is high,  three variance/covariance entries are the most significant, implying these three pairs are the most correlated. The results are consistent with the ``true'' O-D demand. On the other hand, if Lasso penalty $\lambda$ is too small, most of the entries are selected and non-zero, which is inconsistent with the ``true'' demand that has 6 zero entries. Only when the Lasso penalty $\lambda\in [0.05, 0.15]$, the variance/covariance matrix can be estimated accurately.

%since there are several zero entries in the O-D variance/covariance matrix, , those entries are not estimated to be zero, only when we apply proper regularization factor $\lambda$, the variance/covariance matrix can be estimated accurately.

% We can also determine the sparsity for our estimation, for example when choosing $\lambda = 0.2$, $2$ entries will be estimated to be $0$.

\begin{figure}[h!]
\centering
\includegraphics[scale = 0.6]{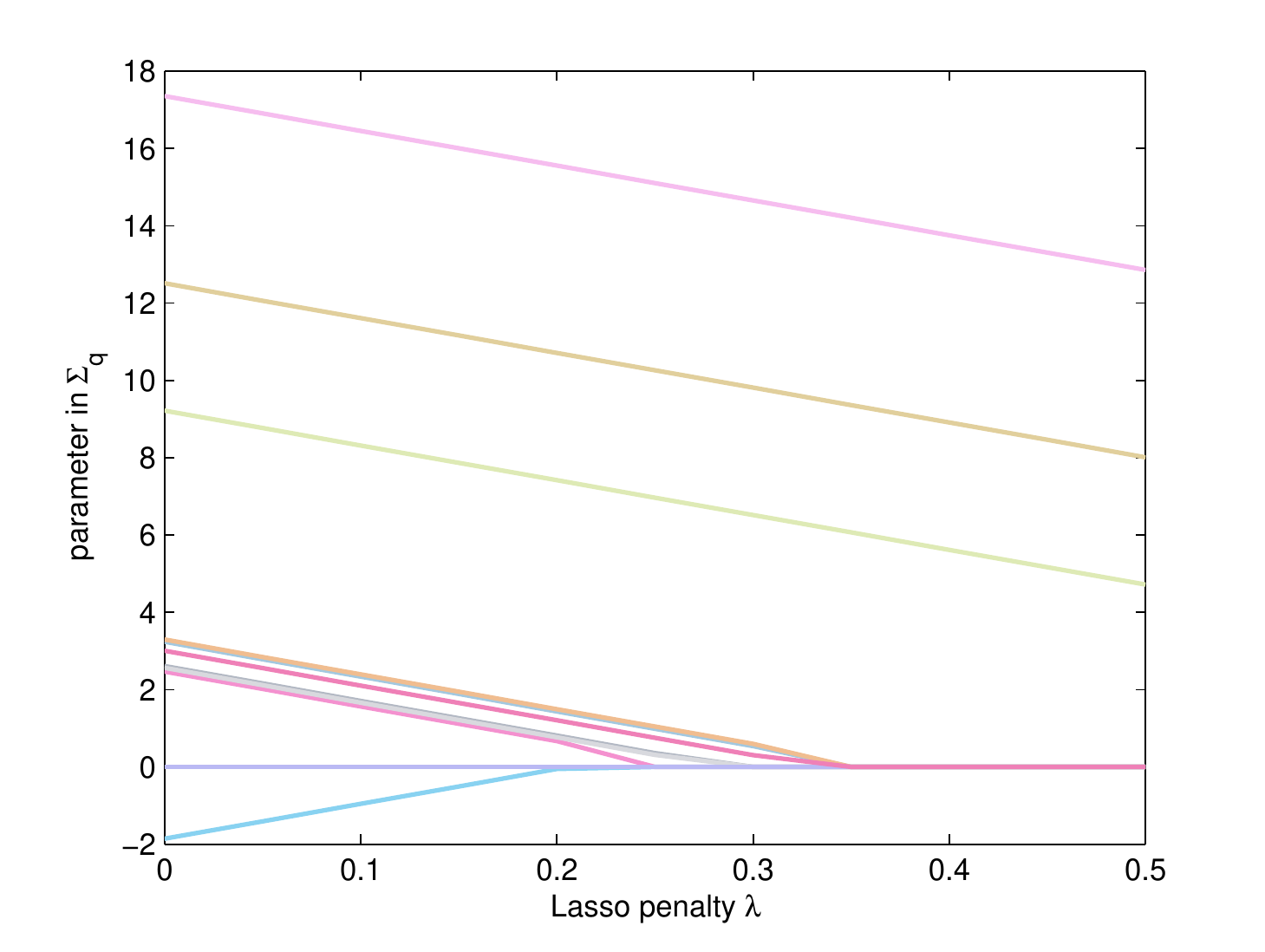}
\caption{\footnotesize Coefficient paths in $\Sigma_q$}
\label{fig:path}
\end{figure}

\subsection{A large-scale network: California SR-41 corridor network}

The proposed framework is now applied to a real-world network to demonstrate its computational efficiency. The SR-41 corridor network is located in the City of Fresno, California. This network consists of one major freeway and two parallel arterial roads connected with local streets. The network is presented in Figure~\ref{fig:network_sr41}, containing 2,413 links and 7,110 O-D pairs. Its O-D demand mean was calibrated by \citet{liu2006streamlined, zhang2008developing}.

\begin{figure}[h!]
\centering
\includegraphics[scale = 0.4]{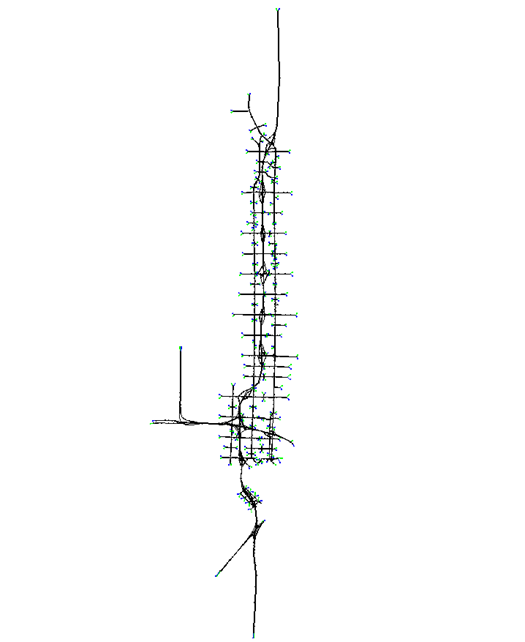}
\caption{\footnotesize The California SR-41 corridor network}
\label{fig:network_sr41}
\end{figure}

We assume that the ``true'' O-D demand variance is the same as its mean (similar to a Poisson distribution), and that $10\%$ of O-D pairs (randomly chosen) are mutually correlated with a correlation randomly drawn from $-0.5$ to $0.5$. We randomly choose $50\%$ of the links on the network to be observed for 1,000 days. Again, a standard BPR function is used for all links. We use Logit-based GESTA as the underlying statistical traffic assignment model, paired with Lasso regularization and historical O-D demand mean. The historical O-D demand variance/covariance matrix is set to the identity matrix. As for the historical O-D demand mean, we uniformly sample a perturbation value from -20\% to 20\%, independently for each O-D pair. It is also used as the initial values for the probabilistic ODE process. The path set is generated by running $3$-shortest paths algorithm for each O-D pair before the ODE process. In all 17,835 paths are considered. The path set for the entire network is assumed to be pre-determined and fixed during this estimation process. The proposed method is developed under MATLAB 2014a and runs on a regular desktop computer (Inter(R) Core i5-4460 3.20 GHz $\times$2, RAM 8 GB).  As a result, the average computation time for one iteration of updating both mean and variance/covariance matrix is $301.82$ seconds. The memory usage over first $10$ IGLS iterations is presented in Figure~\ref{fig:memory}. In general, the sub-problem of estimating O-D demand mean consumes less memory than the sub-problem of estimating the variance/covariance matrix. Peak memory usage is around $4.5$GB for this network settings. The memory usage is closely related to the sparsity level of the O-D variance/covariance matrix.

\begin{figure}[h!]
\centering
\includegraphics[scale = 0.5]{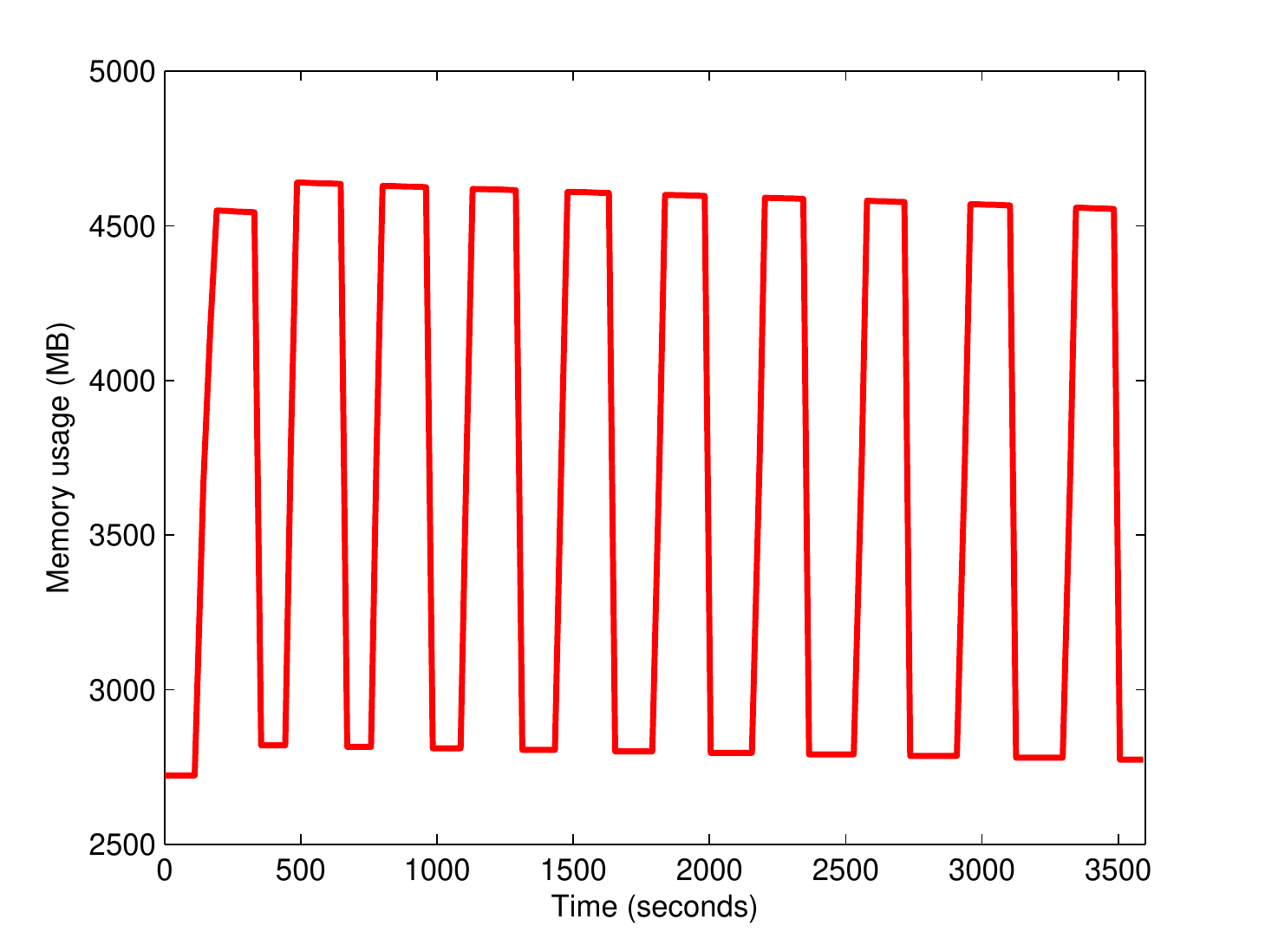}
\caption{\footnotesize SR-41 network: memory usage over first $10$ IGLS iterations}
\label{fig:memory}
\end{figure}

We perform $99$ iterations for the entire IGLS framework. Under each IGLS iteration, we perform $9$ iterations for each of the sub-problems. The convergence of both O-D demand mean and variance/covariance matrix is presented in Figure~\ref{fig:sr41iteration}, at the level of iterations for sub-problems. As can be seen, both sub-problems can be solved very efficiently. The entire process of 900 iterations takes $486$ minutes, but the estimate is reasonably good within approximately $300$ minutes. In addition, we plot the estimated path flow mean against ``true'' flow mean, as well as estimated O-D demand variance/covariance against ``true'' variance/covariance in Figure~\ref{fig:sr41accuracy}. Figure~\ref{fig:sr41ove} plots and the estimated link flow mean against ``true'' flow mean, as well as the estimated link flow variance, against ``true'' variance/covariance of the marginal distributions of link flow. The proposed probabilistic ODE seems computationally plausible on a sizable network and is able to achieve reasonably accurate results, approaching the synthesized ``true'' day-to-day demand mean and covariance.

\begin{figure}[h!]
\centering
\includegraphics[scale = 0.35]{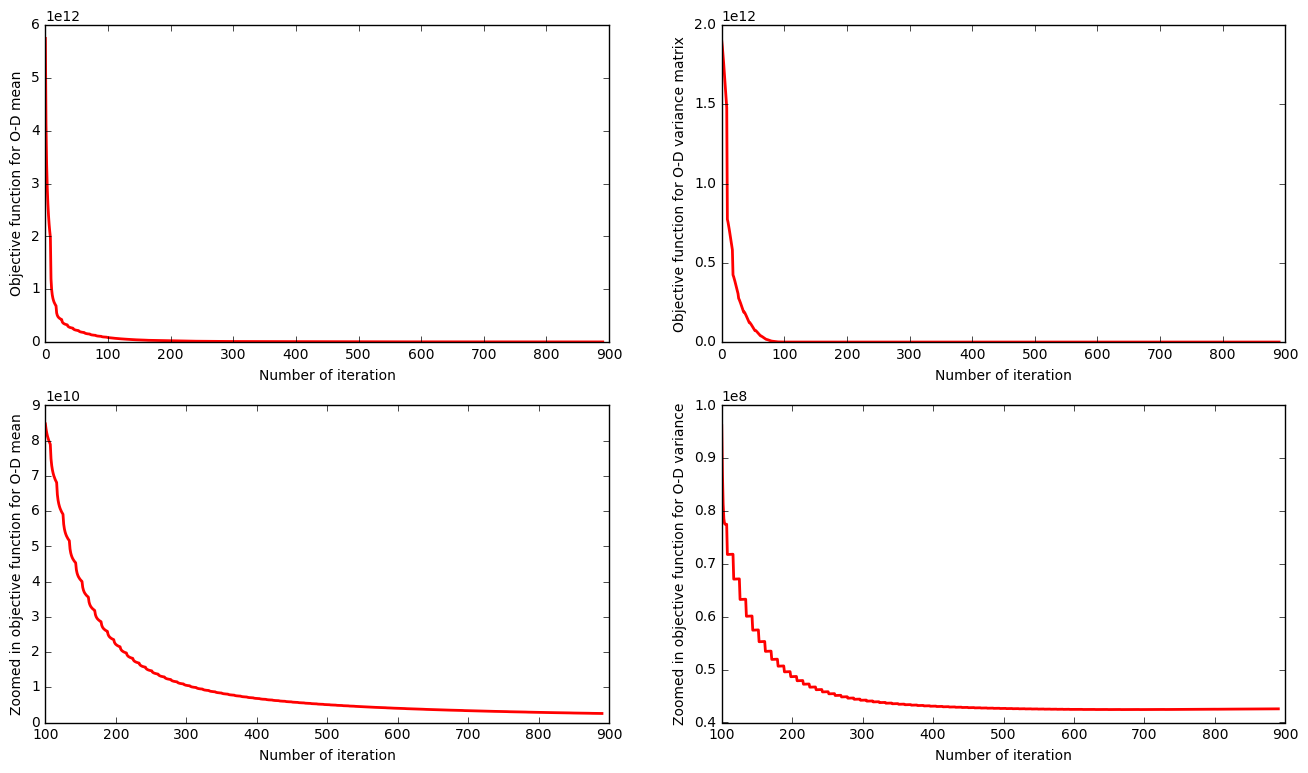}
\caption{\footnotesize Convergence for both O-D demand mean and covariance matrix}
\label{fig:sr41iteration}
\end{figure}

\begin{figure}[h!]
\centering
\includegraphics[scale = 0.55]{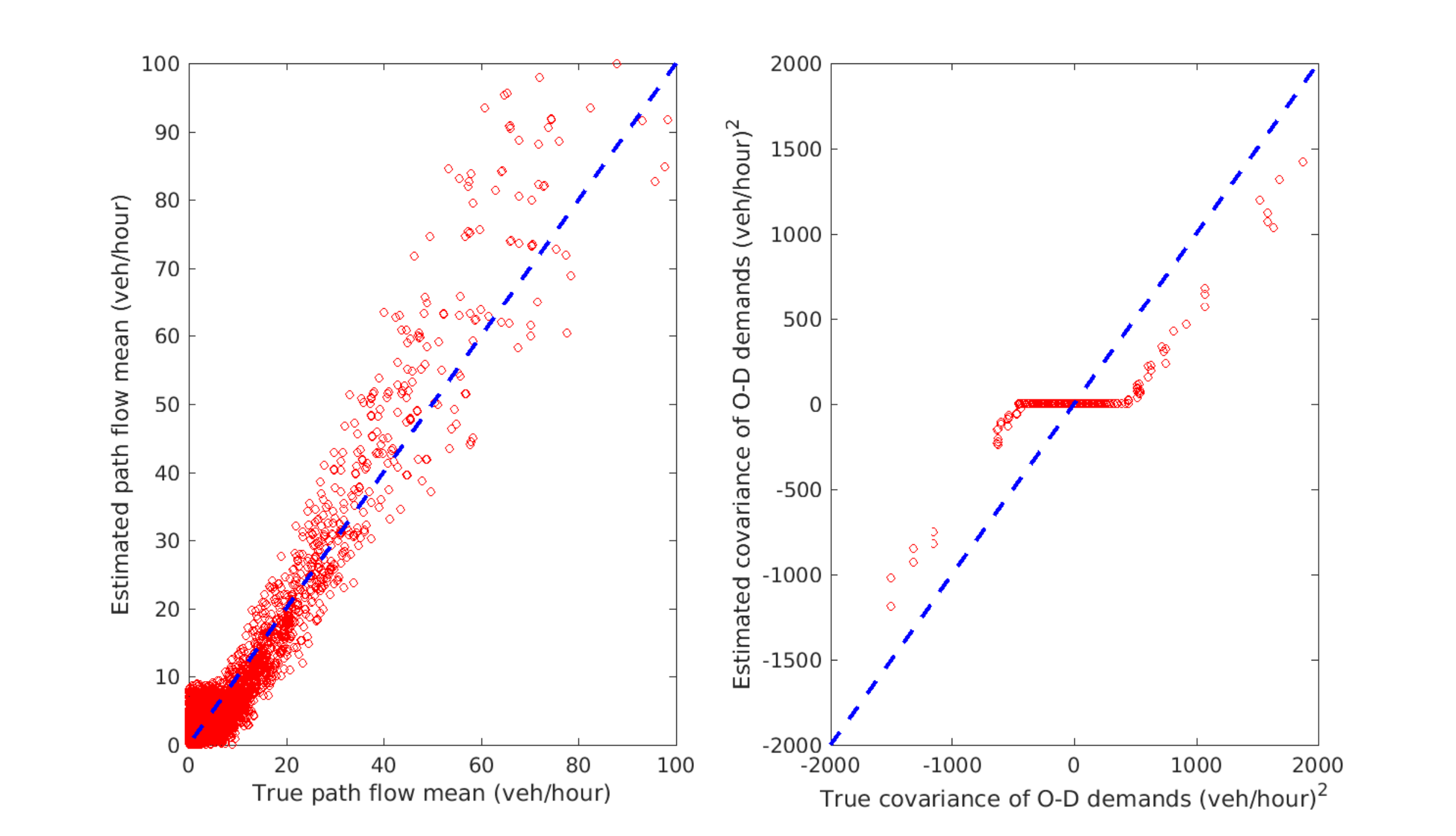}
\caption{\footnotesize Estimated and ``true'' path flow mean and O-D demand variance}
\label{fig:sr41accuracy}
\end{figure}

\begin{figure}[h!]
	\centering
	\includegraphics[scale = 0.5]{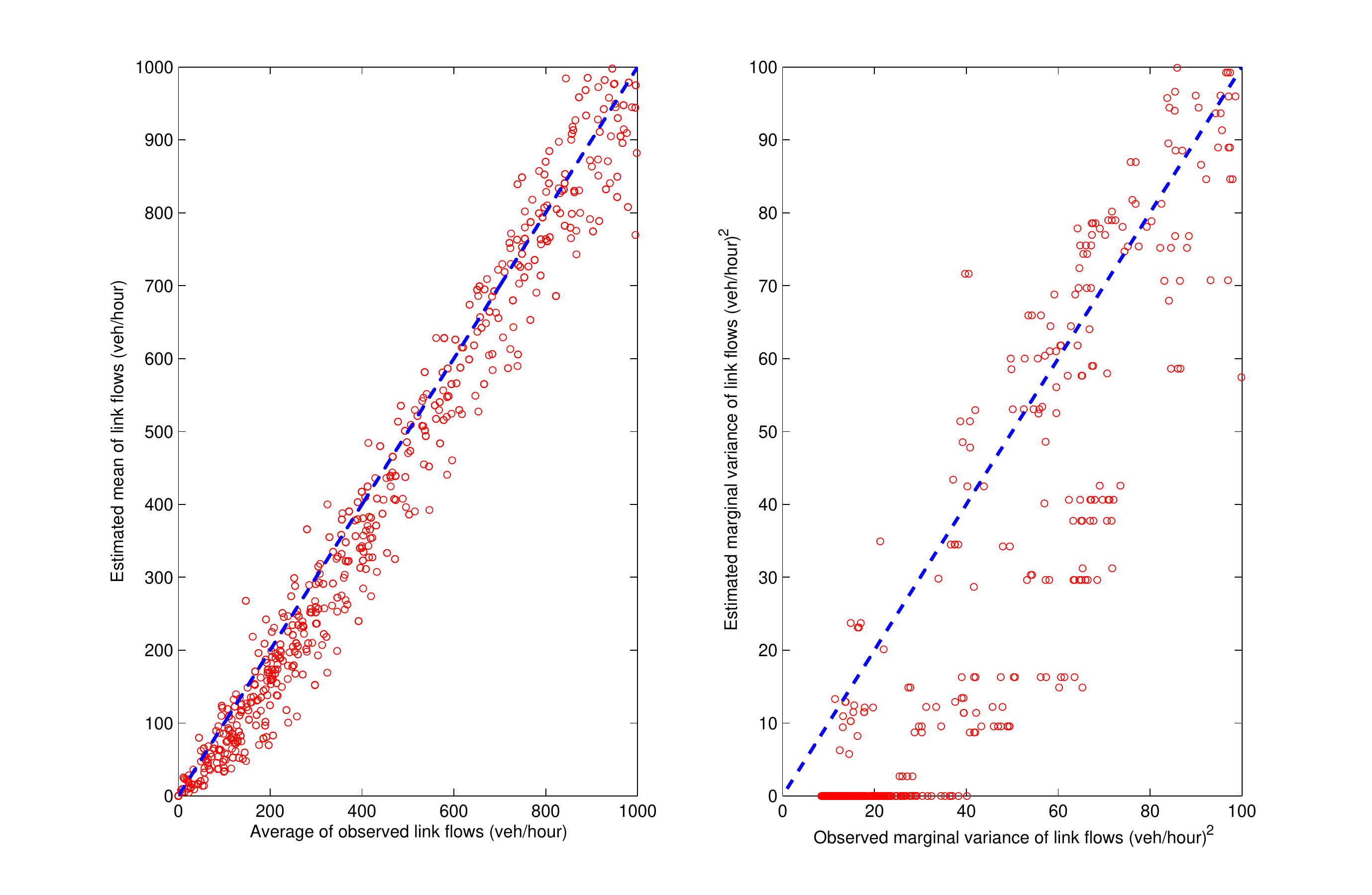}
	\caption{\footnotesize Estimated and ``true'' link flow (Left: mean; Right: variance of the marginal distributions)}
	\label{fig:sr41ove}
\end{figure}

One interesting result is that the Lasso regularization on the demand variance/covariance matrix unavoidably leads to a biased estimation, as can be seen from both Figure~\ref{fig:sr41accuracy} and Figure~\ref{fig:sr41ove}. Most of the variance/covariance (for both O-D demand and link flow) are either estimated as zeros, or substantially underestimated, as a result of Lasso shrinking. However, without the Lasso regularization, most of those variance/covariance entries in the matrix that are substantial can go way off the chart. Again, a proper Lasso penalty ensures a good trade-off between bias and variance of the estimation. In practice, trial-and-error may be needed to identify a proper Lasso penalty.

We also conduct another experiment for this network setting with positively correlated O-D demand, namely 10\% of O-D pairs (randomly chosen) are mutually correlated with a correlation randomly drawn from $0$ to $0.5$. All other settings are the same as before. This experiment aims at examining the robustness of the proposed probabilistic OD estimator. The estimation results are presented in Figure~\ref{fig:sr41again}. Clearly the proposed method accurately estimates the probabilistic O-D demand in terms of both the mean and variance-covariance matrix.

\begin{figure}[h!]
	\centering
	\includegraphics[scale = 0.55]{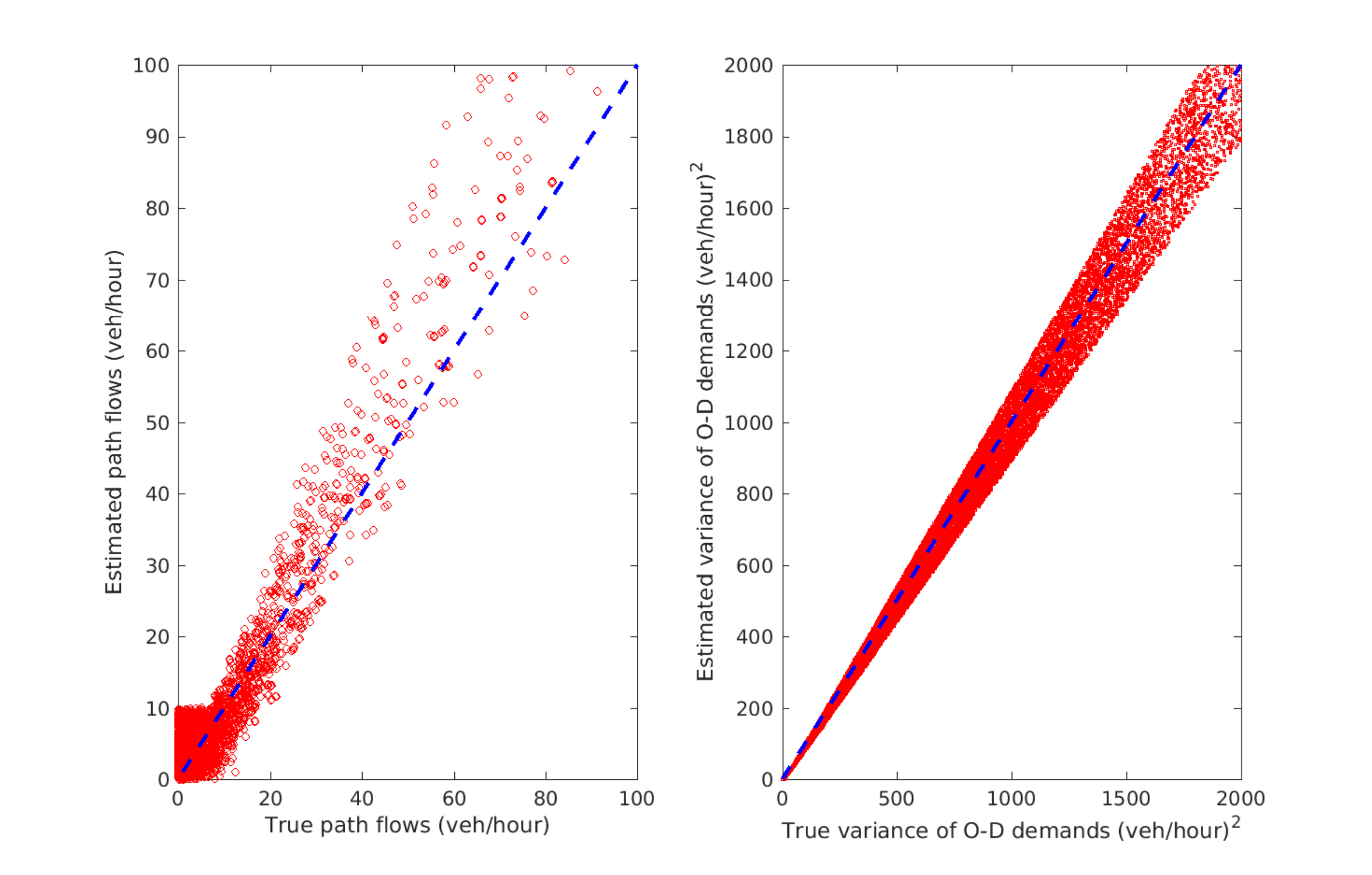}
	\caption{\footnotesize Estimated and ``true'' path flow mean and O-D demand variance when all true correlations are positive}
	\label{fig:sr41again}
\end{figure}

\subsection{A second large-scale network: Washington D.C. Downtown Area}
Previous experiments are conducted in a simulated environment where observations data were synthesized. In this subsection, we apply our probabilistic ODE to a real world network in Washington D.C. by using the actual day-to-day traffic count data. %, which is drawn by ArcMap\footnote{\url{http://desktop.arcgis.com/en/arcmap/}}.

\begin{figure}[h!]
\centering
\includegraphics[scale = 0.4]{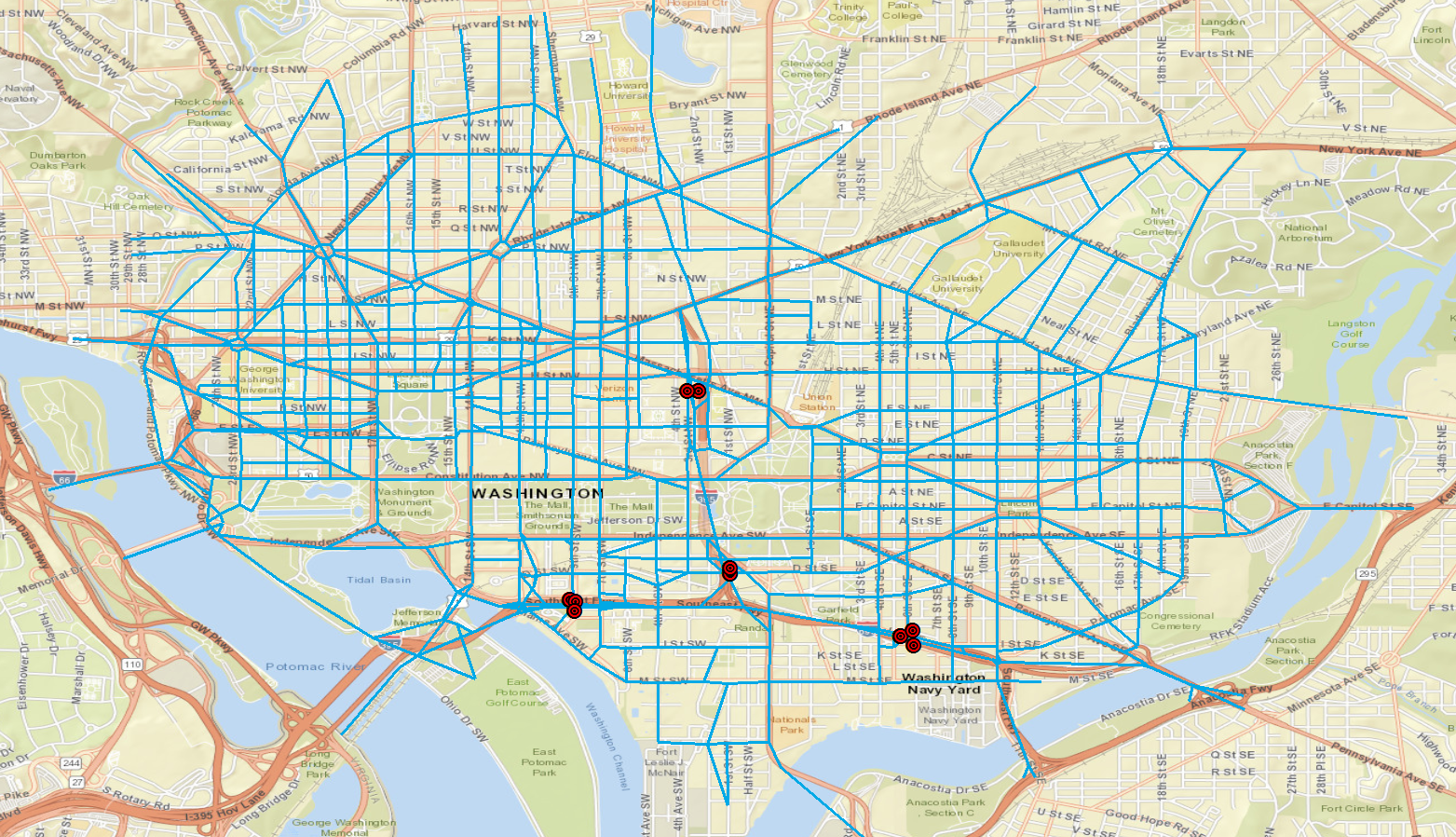}
\caption{\footnotesize The Washington D.C. Downtown network}
\label{fig:dc}
\end{figure}

This network is generally a grid network that consists of $984$ road junctions, 2,585 road segments and 4,900 O-D pairs, The overview of the network is shown in Figure~\ref{fig:dc}. Red dots on the map represent those active fixed-location sensors. There are in all $51$ sensors in the region, while only $10$ of those sensors are working under healthy conditions and collecting data continuously from 2008. We obtain the traffic counts data from August 2008 through December 2015. Around 2,000 data samples were observed for each sensor. We box plot the aggregated traffic counts of all sensors during the morning peak over the years (four of them are presented in Figure~\ref{fig:dc_count}). The box plot shows that the mean and variance of each sensor do not change as much over the $7$ years. Thus, we decide to use all the days for each sensor to estimate the probabilistic O-D demand representing the demand over a course of 7 years.

\begin{figure}[h!]
\centering
\includegraphics[scale = 0.5]{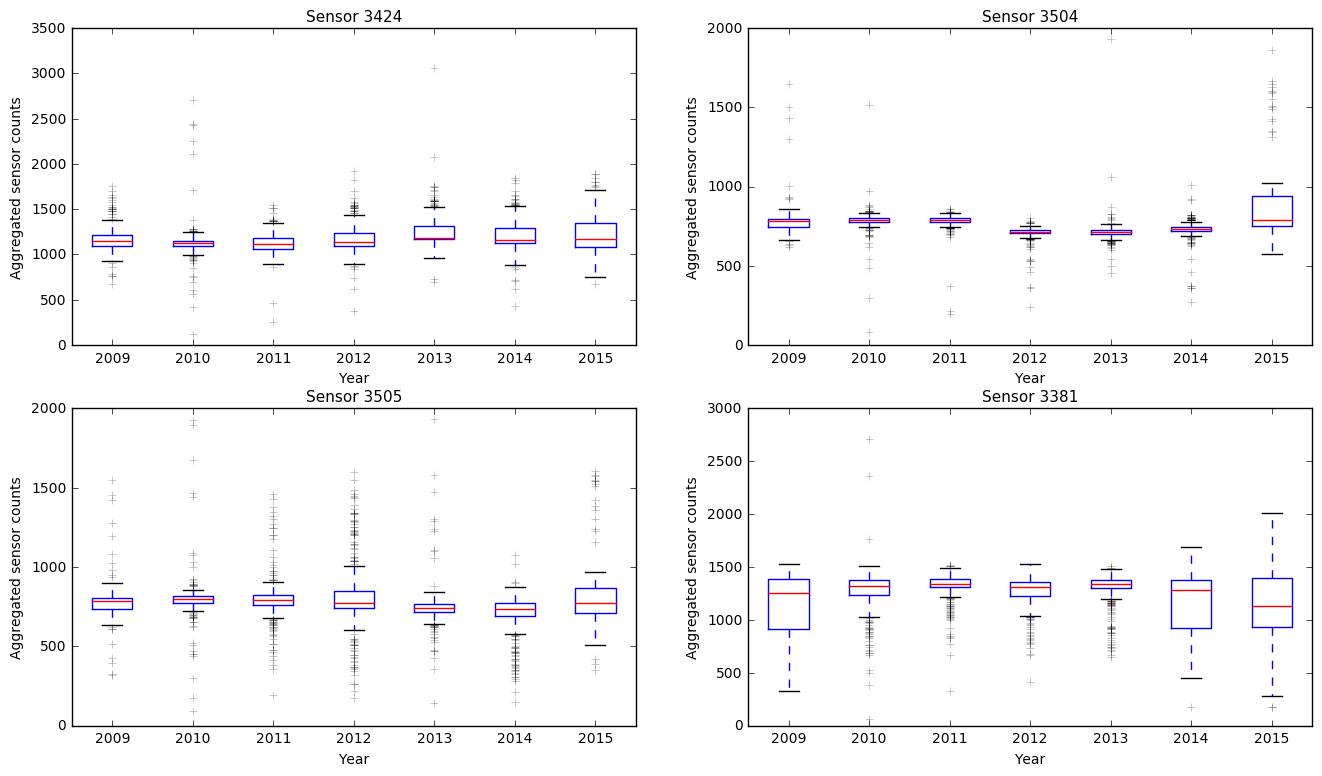}
\caption{\footnotesize Aggregated traffic counts during the morning peak for four selected sensors from 2009 to 2015}
\label{fig:dc_count}
\end{figure}

We again use Logit-based GESTA as the underlying statistical traffic assignment model, paired with Lasso regularization and historical O-D demand mean. The historical O-D demand variance/covariance matrix is set to the identity matrix. The historical O-D demand mean is obtained from the planning model of Year 2013 developed by Metropolitan Washington Council of Governments, which is also used as the initial demand mean for the solution process. The initial demand variance/covariance matrix is randomly generated.

The convergence of both sub-problems are presented in Figure~\ref{fig:dc_converge}. Overall, the solution algorithm performs well on both sub-problems. It takes $40.18$ minutes to complete all 900 iterations using the same programming environment and aforementioned computer. The estimated O-D demand, under the probabilistic ODE framework, is able to reproduce the day-to-day observations on those observed links,  as shown in Figure~\ref{fig:dc_result}. Again, the variance of the marginal distributions of most observed links is underestimated (or estimated as zeros) as expected, due to Lasso regularization. In addition, the estimation seems robust to a few outliers identified in Figure~\ref{fig:dc_count}. Overall, the results of the probabilistic ODE are compelling and satisfactory. However, we speculate that the initial variance/covariance matrix can be critical to the final estimation results. Some prior knowledge about the variance/covariance of demand among O-D pairs can be obtained from traditional planning models, which may help improve the estimation results for real-world networks.

\begin{figure}[h!]
	\centering
	\includegraphics[scale = 0.5]{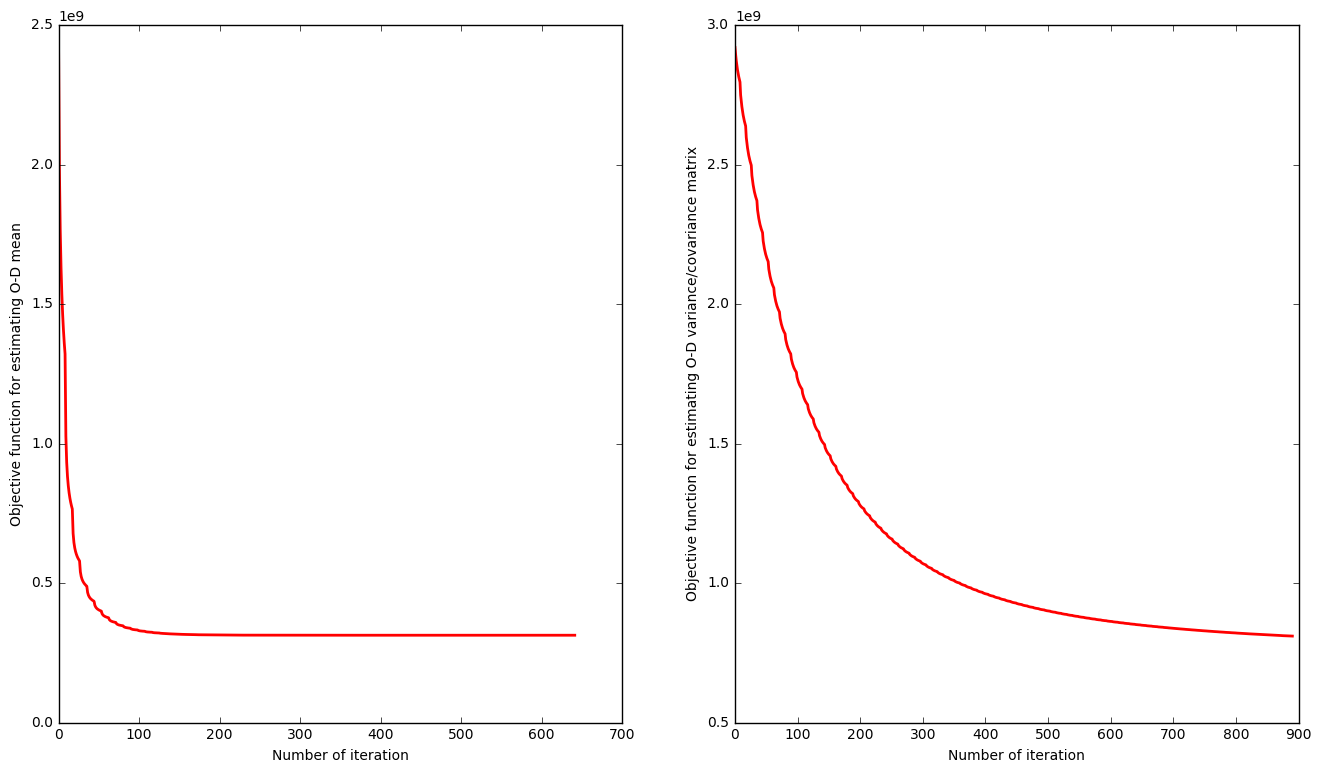}
	\caption{\footnotesize Convergence of both sub-problems for the D.C network}
	\label{fig:dc_converge}
\end{figure}

%The convergence curves are not very smooth since a better estimation of O-D mean will also reduce the objective value in estimating O-D variance/covariance matrix. The two sub-problem accelerates each other during the overall IGLS iterations.

\begin{figure}[h!]
	\centering
	\includegraphics[scale = 0.3]{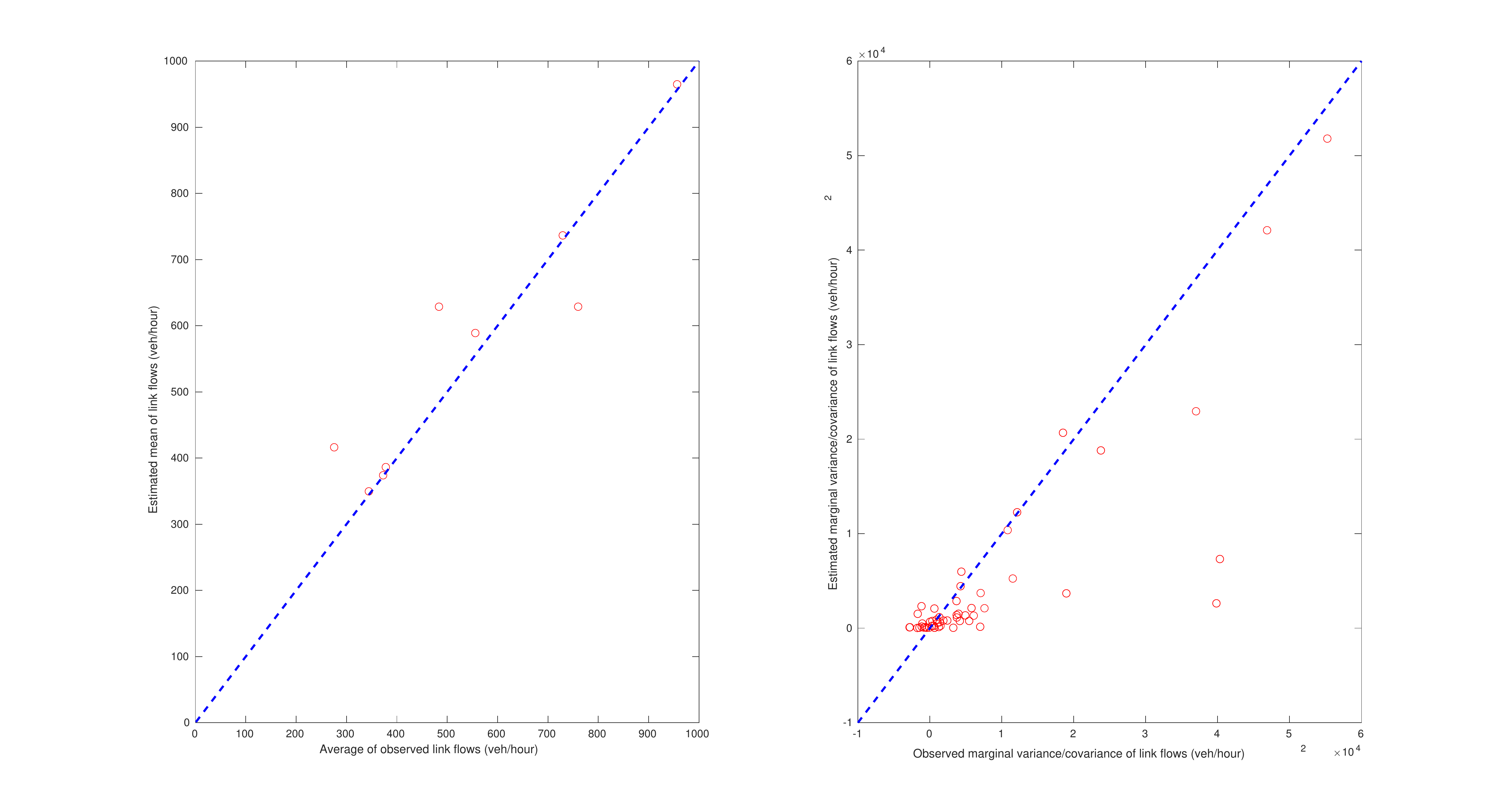}
	\caption{\footnotesize Estimated and observed link flow during the morning peak (Left: mean; Right: variance/covariance)}
	\label{fig:dc_result}
\end{figure}

\section{Conclusions}
\label{sec:con}
This paper develops a novel theoretical framework for estimating the mean and variance/covariance matrix of O-D demand considering the day-to-day variation induced by travelers' independent route choices. The essential idea is to see the traffic data on each day as one data point and to use data points collected years along to estimate the probability distribution of O-D demand, as well as the probability distributions of links, paths and their generalized costs. As opposed to a real-valued estimation of flow and costs from traditional ODE, the probabilistic ODE estimates their probability distributions that are central to reliable network design, operation and planning. The probabilistic ODE framework is large-scale data friendly in the sense that it can make the best use of large-scale day-to-day traffic data to support complex decision making.

The framework estimates O-D mean vector and variance/covariance matrix iteratively, also known as iterative generalized least squares (IGLS) in statistics. IGLS holds great potential to converge faster than a formulation that estimates both simultaneously. It also decomposes a complex estimation problem into two sub-problems, which are relatively easier to solve. In the sub-problem of estimating O-D demand mean, we illustrate how to incorporate day-to-day traffic flow observations into the formulations and explain how the data size and historical O-D information effect the estimation results. In the sub-problem of estimating the O-D demand variance-covariance matrix, a convex optimization formulation is presented to approximate the solution. Lasso regularization is employed to obtain sparse covariance matrix for better interpretation and computational efficiency. We also discuss the observability of the probabilistic ODE problem. The non-uniqueness property of the probabilistic ODE under the IGLS framework is examined.  Though probabilistic ODE works with a much larger solution space than the deterministic O-D ODE,
we show that its estimator for O-D demand mean is no worse than the best possible estimator by an error that reduces with the increase in sample size. %Different strategies to mitigate the non-uniqueness are proposed for both single-level and bi-level O-D demand estimation problems. The observability of bi-level O-D estimation  under user equilibrium and stochastic user equilibrium constraints are also discussed separately.

The probabilistic ODE is examined on two small networks and two real-world large-scale networks. The solution converges quickly under the IGLS framework. In all those experiments, the results of the probabilistic ODE are compelling, satisfactory and computationally plausible. We also conduct the sensitivity analysis of estimation performance with respect to data sample size and historic O-D information. Increase the sample size and quality of historical O-D information can effectively approach the ``true'' probability distribution of O-D demand and path/link flow. Lasso regularization on the covariance matrix estimation leans to underestimate variance and covariance. A proper Lasso penalty ensures a good trade-off between bias and variance of the estimation. In practice, trial-and-error may be needed to identify a proper Lasso penalty.

In the near future, we plan to address a few computational issues before it can be widely deployed for practitioners. We will intensively test this probabilistic ODE method in other large-scale networks with a better data coverage than the D.C. network tested in this paper. Various modeling settings need to be tested, such as different route choice models, with and without Lasso regularization, with and without historical O-D, with and without traffic speed data. In addition, we speculate that the initial variance/covariance matrix can be critical to the estimator. Some prior knowledge about the variance/covariance of demand among O-D pairs can be obtained from traditional planning models. We plan to test how this prior knowledge can help improve the estimation results. In addition, this paper assumes the traffic observations on a set of days are i.i.d, but the set can be flexible. We can fit two probabilistic O-D distributions using solely workdays and weekends. We can even construct an unsupervised learning mechanism to cluster the traffic observations and then fit probabilistic O-D distribution for each cluster.

Furthermore, we plan to extend this research to estimate the probability distributions of time-varying O-D demand where mesoscopic traffic flow dynamics can be incorporated into the network modeling instead of naive BPR functions.

\section*{Acknowledgement}
This research is funded in part by Traffic 21 Institute and Carnegie Mellon University's Mobility21, a National University Transportation Center for Mobility sponsored by the US Department of Transportation. The contents of this report reflect the views of the authors, who are responsible for the facts and the accuracy of the information presented herein. The U.S. Government assumes no liability for the contents or use thereof. We would also like to thank anonymous reviewers for their valuable suggestions.

\cleardoublepage
\bibliography{report}
\cleardoublepage

\appendix
\section{Estimating the link flow on observed links}
\label{ap:ob}
Identifying the estimator for the link flow on observed links can be cast into a quadratic optimization problem as shown in Equation~\ref{eq:rewrite}.
%\begin{equation}
%\label{eq:rewrite}
%\begin{array}{rrclcl}
%\vspace{5pt}
%\displaystyle \min_{x} & \multicolumn{3}{l}{\displaystyle  \frac{1}{2} x^T Q x +  b^T x} \\
%\textrm{s.t.} & x & \geq & 0 \\
%\end{array}
%\end{equation}

%Additionally we need to add the non-negativity constrain on $x$, so the final estimation process can be cast to a quadratic optimization problem:
\begin{equation}
\label{eq:rewrite}
\begin{array}{rrclcl}
\vspace{5pt}
\displaystyle \min_{x} & \multicolumn{3}{l}{\displaystyle  \frac{1}{2} x^T B x +  b^T x} \\
\textrm{s.t.} & x & \geq & 0 \\
\end{array}
\end{equation}
where
\begin{eqnarray}
B &=&n   \Sigma_{x}^{-1}\\
b &=& - \sn  \Sigma_x^{-1} \xx_i
\end{eqnarray}

First the above optimization problem is strongly convex, since $\Sigma_x =  \Delta^T \Sigma_{f|q} \Delta  + \Delta P \Sigma_{q} \Delta^T P^T + \Sigma_e$ and $\Sigma_e$, $\Sigma_{f|q}$ semi-positive definitive and $\Sigma_e$ is positive definitive. $\Sigma_x \succ 0$ always holds, and so does $\Sigma_x^{-1} \succ 0$

Proposition \ref{pro:mle} implies that $\hat{x^o}$ can be derived in a closed form.  This is proven by using the KKT condition of the minimization problem above.
%\begin{proposition}[MLE]
%\label{pro:mle}
%Given the data set $\xx^o$ and link flow covariance matrix $\Sigma_x^o$ on observed links, the estimator is $\hat{x}^o = \frac{1}{n}   \xx_i^o$.
%\end{proposition}
%\begin{proposition}
%Maximum likelihood estimator of $x^o$ can be derived in closed form:
%\begin{eqnarray}
%\hat{x^o} = \frac{1}{n}   \xx_i^i
%\end{eqnarray}
%\end{proposition}
\begin{proof}
%First we write down the KKT condition of the quadratic optimization.
\begin{eqnarray}
\sn {\Sigma_x^{o}}^{-1} (x^o - \xx_i^o) - \sum_{a=1}^{|A^o|} \lambda_a &=& 0\\
\lambda_a &\geq& 0\\
\lambda_a x_i^o &=& 0\\
x_i^o &\geq& 0
\end{eqnarray}
Since $\xx_i^o \succ 0$, $x^o = \frac{1}{n} x_i > 0$. Then, $\lambda_a = 0$, all four KKT conditions are satisfied. Since the minimization problem is strongly convex, therefore $\hat{x^o} = \frac{1}{n}   \xx_i^o$ is the optimal solution to the optimization problem.
\end{proof}
\cleardoublepage
\section{Estimating the covariance matrix with Lasso regularization}
\label{ap:proximal}
We discuss the procedure for estimating the covariance matrix with Lasso regularization. We first use the proximal method Iterative Shrinkage-Thresholding Algorithm (ISTA) \citep{nesterov1983method, beck2009fast}. ISTA can be modified to the Fast Iterative Shrinkage-Thresholding Algorithm (FISTA) \citep{nesterov2005smooth}. FISTA's convergence rate is of $O(\frac{1}{n})$, compared to $O(\frac{1}{n^2})$ for ISTA.

Rewrite Formulation \ref{eq:LSwithLasso}.
\begin{equation}
\begin{array}{rrclcl}
\vspace{5pt}
\displaystyle \min_{\Sigma_q} & \multicolumn{3}{l}{\displaystyle   f(\Sigma_q) +  \lambda  g(\Sigma_q)}\\
\textrm{s.t.} & \Sigma_q & \succeq & 0
\end{array}
\end{equation}
where
\begin{eqnarray}
f(\Sigma_q) &=& \norm{\Delta \Sigma_{f|q} \Delta^T +\Delta P \Sigma_q P^T \Delta^T +  \Sigma_e - \Sigma_x^o}_F^2 \\
&=& \text{Tr}\left[\left(\Delta \Sigma_{f|q} \Delta^T +\Delta P \Sigma_q P^T \Delta^T - \Sigma_x^o\right)\left(\Delta \Sigma_{f|q} \Delta^T +\Delta P \Sigma_q P^T \Delta^T +  \Sigma_e - \Sigma_x^o\right)^T \right] \\
&=& \text{Tr}\left[\Delta P \Sigma_q P^T \Delta^T \left( 2\Delta \Sigma_{f|q} \Delta^T + \Delta P \Sigma_q P^T \Delta^T  + 2\Sigma_x^o   \right) \right]  + \mathcal{C} \\
\end{eqnarray}
where $\mathcal{C} $ is a constant independent of $\Sigma_q$.

Using results presented in \citet{petersen2008matrix}, the derivative of $f(\Sigma_q)$,
\begin{eqnarray}
\frac{\partial f(\Sigma_q)}{\partial \Sigma_q} &=& 2 P^T \Delta^T \left(\Delta \Sigma_{f|q} \Delta^T  - \Sigma_x^o\right)\Delta P  + 2 P^T \Delta^T \Delta P \Sigma_q P^T \Delta^T \Delta P \\
&=&2  P^T \Delta^T \left(\Delta P \Sigma_q P^T \Delta^T + \Delta \Sigma_{f|q} \Delta^T - \Sigma_x^o\right)\Delta P
\end{eqnarray}
Then define the soft-thresholding operator $S_\lambda(\beta)$:
\begin{eqnarray}
S_\lambda(\beta) = \begin{cases}
\beta_i - \lambda  & \text{if $\beta_i > \lambda$} \\
0  & \text{if $|\beta_i| \leq \lambda$} \\
\beta_i + \lambda  & \text{if $\beta_i <  -\lambda$}
\end{cases}
\end{eqnarray}
ISTA's updating rule is:
\begin{eqnarray}
\Sigma_q^+ =  S_\lambda\left(\Sigma_q - \frac{\partial f(\Sigma_q)}{\partial \Sigma_q}\right)
\end{eqnarray}
FISTA's updating rule is:
\begin{eqnarray}
v &=& \Sigma_q^{(k-1)} + \frac{k-2}{k+1} \left(\Sigma_q^{(k-1)} - \Sigma_q^{(k-2)}\right)\\
\Sigma_q^{(k)} &=& S_\lambda\left(  v - \frac{\partial f(\Sigma_q)}{\partial \Sigma_q}\right)
\end{eqnarray}

When applying the updating rules, proper step sizes need to be chosen to ensure $\Sigma_q  \succeq  0$, backtracking methods can be used to select the step sizes \citep{boyd2004convex}. Convergence analysis for both ISTA and FISTA can be found in \citet{nesterov1983method, nesterov2005smooth}.

\cleardoublepage
\section{Proof of Proposition~\ref{pro:better}}
\label{ap:risk}

Denote $\X$ as the sample space where we observe the data from, $x$ is one possible observation from the sample space. The ODE method $d(x)$ takes the observations $x$ as input and outputs the O-D estimators. The loss function $L(\cdot, \cdot)$ measures the discrepency between the estimated value and the true value. The risk is defined as the expectation of the loss between the true O-D demand mean $q$ and the estimated O-D demand mean by all possible observations. The intuition to define risk is that, we want to minimize the loss for all possible observations, rather than one specific observation (the empirical loss for specific numerical experiments). ODE methods that minimize the risk will have robust performance, even provided with noisy and limited data inputs. The risk function is defined by,
\begin{eqnarray}
R(q, d)   = \int_{\X} L(q, d(x))  p_{Q}(x) dx
\end{eqnarray}

The loss function can be any non-negative and strictly convex function. In this study, we use the quadratic loss since it is the most commonly used. However, the following proof would generally work for other norm operators.
\begin{eqnarray}
L(x, y) = \norm{x-y}_2^2
\end{eqnarray}

We first consider Formulation~\ref{eq:traOD} without historical O-D information,
\begin{equation}
\label{eq:sod}
\begin{array}{rlclcl}
\vspace{5pt}
\displaystyle \min_{q, x, f} & \frac{1}{2} (x^o - \Delta^o \tilde{p} q)^T \hat{\Sigma}_x^{-1} (x^o - \Delta^o \tilde{p} q) \\
\textrm{s.t.} & q \geq 0\\
\end{array}
\end{equation}
where $\hat{\Sigma}_x^{-1}$ is the estimated inverse link variance/covariance matrix. $\hat{\Sigma}_x^{-1}$ can be estimated during the IGLS iteration, we rewrite the formulation in the form of $L^2$ norm, % shown in Equation~\ref{eq:sod2}.
\begin{eqnarray}
\label{eq:sod2}
\begin{array}{rlclcl}
\vspace{5pt}
\displaystyle \min_{q, x, f} & \frac{1}{2} \norm{ \hat{\Sigma}_x^{-\frac{1}{2}}x^o - \hat{\Sigma}_x^{-\frac{1}{2}}\Delta^o \tilde{p} q}_2^2  \\
\textrm{s.t.} & q \geq 0\\
\end{array}
\end{eqnarray}

Formulation~\ref{eq:sod2} is a standard non-negative least square problem. Many efficient algorithms \citep{lawson1995solving} are practically ready to solve the formulation. Here we propose a fairly good estimation of O-D demand mean $q$ in Lemma~\ref{lem:q} rather than directly solving for Formulation~\ref{eq:sod2}. We will later see that the proposed estimator $\hat{q}$ approximates the true O-D mean $q$ when the data size $n$ is sufficiently large.
\begin{lemma}
\label{lem:q}
When the data size $n$ is sufficiently large, the O-D demand mean can be estimated by,
\begin{eqnarray} \label{sampledx}
\hat{q} = \max\left(\tilde{D}^+  v^o, 0\right)
\end{eqnarray}
where $v^o = \hat{\Sigma}_x^{-\frac{1}{2}} x^o$, $\tilde{D} = \hat{\Sigma}_x^{-\frac{1}{2}}\Delta^o p$ and $\tilde{D}^+$ is the Moore–Penrose pseudoinverse of matrix $\tilde{D}$. Detailed proof can be found in \citet{nie2005inferring}.
\end{lemma}

If $|K| \leq |A^o|$, then $\tilde{D}^+ = \tilde{D}^{-1}$. Otherwise, $\tilde{D}^+ = (\tilde{D}^T \tilde{D})^{-1} \tilde{D}^T$. Based on Lemma~\ref{lem:q}, we prove Proposition~\ref{pro:better}.

\begin{proof}
Our first target is to bound the risk when the data sample size increases. The risk for a given estimation method $d(x)$ presented in \ref{sampledx} is,
\begin{eqnarray}
R(q, d) &=& \int_{\X} L\left(q, d(x) \right)   p_{Q}(x) dx\\
&=& \int_{\X} \norm{q - \max\left(\tilde{D}^+ v^o , 0\right)}_2^2  p_{Q}(x) dx\\
&\leq&\int_{\X} \norm{q - \tilde{D}^+ v^o }_2^2  p_{Q}(x) dx\\
&=& \int_{\X} \norm{q - \left((\hat{\Sigma}_x^{-\frac{1}{2}}\Delta^o \tilde{p})^T(\hat{\Sigma}_x^{-\frac{1}{2}}\Delta^o \tilde{p})\right)^{-1} (\hat{\Sigma}_x^{-\frac{1}{2}}\Delta^o \tilde{p})^T v^o }_2^2  p_{Q}(x) dx\\
&=& \int_{\X} \norm{q - \left(\tilde{p}^T{\Delta^o}^T\hat{\Sigma}_x^{-1}\Delta^o \tilde{p}\right)^{-1} \tilde{p}^T {\Delta^o}^T \hat{\Sigma}_x^{-1} x^o }_2^2  p_{Q}(x) dx \label{eq:inv}\\
&\leq&\int_{\X} \norm{\left(\tilde{p}^T{\Delta^o}^T\hat{\Sigma}_x^{-1}\Delta^o \tilde{p}\right)^{-1}}_2^2 \norm{ \left( \tilde{p}^T {\Delta^o}^T \hat{\Sigma}_x^{-1} \Delta^o \tilde{p}  \right)q -  \tilde{p}^T {\Delta^o}^T \hat{\Sigma}_x^{-1} x^o }_2^2  p_{Q}(x) dx\\
&\leq&\int_{\X} \norm{\left(\tilde{p}^T{\Delta^o}^T\hat{\Sigma}_x^{-1}\Delta^o \tilde{p}\right)^{-1}}_2^2 \norm{  \tilde{p}^T {\Delta^o}^T \hat{\Sigma}_x^{-1} x  -  \tilde{p}^T {\Delta^o}^T \hat{\Sigma}_x^{-1} x^o }_2^2  p_{Q}(x) dx\\
&=&\int_{\X} \norm{\left(\tilde{p}^T{\Delta^o}^T\hat{\Sigma}_x^{-1}\Delta^o \tilde{p}\right)^{-1}}_2^2 \norm{  \tilde{p}^T {\Delta^o}^T (\hat{\Sigma}_x^{-1} x -   \hat{\Sigma}_x^{-1} x^o) }_2^2  p_{Q}(x) dx\\
&=&\int_{\V} \norm{\left(\tilde{p}^T{\Delta^o}^T\hat{\Sigma}_x^{-1}\Delta^o \tilde{p}\right)^{-1}}_2^2 \norm{  \tilde{p}^T {\Delta^o}^T (v - v^o) }_2^2  p_{Q}(v) dv\\
&=& \norm{\left(\tilde{p}^T{\Delta^o}^T\hat{\Sigma}_x^{-1}\Delta^o \tilde{p}\right)^{-1}}_2^2  \norm{p^T {\Delta^o}^T}_2^2 \Expect \norm{V^o-v^o}_2^2
\end{eqnarray}

In Equation~\ref{eq:inv}, $\tilde{p}^T{\Delta^o}^T\hat{\Sigma}_x^{-1}\Delta^o \tilde{p}$ is invertible when $\Delta^o$ is fully ranked. Since $\norm{p^T {\Delta^o}^T}_2^2 $ is independent of the data size $n$, we can see it as a constant. As for $\norm{\left(\tilde{p}^T{\Delta^o}^T\hat{\Sigma}_x^{-1}\Delta^o \tilde{p}\right)^{-1}}_2^2$,when the sample size increases, $\hat{\Sigma}_x$  approximates $\Sigma_x$. As long as the observed data $x$ is bounded, $\hat{\Sigma}_x$ can be bounded, independent of the sample size $n$. For $v^o$, in the sub-problem of estimating the O-D mean vector, we have,
\begin{eqnarray}
V^o = \hat{\Sigma}_x^{-\frac{1}{2}} \bar{\xx}^o \sim \N(v^o,  \frac{\hat{\Sigma}_x^{-\frac{1}{2}} \Sigma_x \hat{\Sigma}_x^{-\frac{1}{2}}}{n} )
\end{eqnarray}

Again $\hat{\Sigma}_x^{-\frac{1}{2}} \Sigma_x \hat{\Sigma}_x^{-\frac{1}{2}}$ can be bounded. When $n \to \infty$, by Law of large number (LLN), we have,
\begin{eqnarray}
 V^o \xrightarrow{Prob} v^o
\end{eqnarray}
Also after assuming $\norm{\left(\tilde{p}^T{\Delta^o}^T\hat{\Sigma}_x^{-1}\Delta^o \tilde{p}\right)^{-1}}_2^2  \norm{\tilde{p}^T {\Delta^o}^T}_2^2 \norm{\hat{\Sigma}_x^{-\frac{1}{2}} \Sigma_x \hat{\Sigma}_x^{-\frac{1}{2}}}_2^2 \leq M $, we have

\begin{eqnarray}
R(q,d) \leq \frac{M}{n}\in {\cal O}\left(\frac{1}{n}\right), \quad \forall p
\end{eqnarray}

This implies that as long as the estimation of $\Sigma_x$ is bounded, any estimator $d(x)$ can achieve the same level of accuracy provided with a sufficiently large $n$.

For Formulation~\ref{eq:ODwithEQ}, suppose we use heuristic methods to solve the bi-level formulation as in \citet{yang1995heuristic}. Each iteration in solving the upper level problem is equivalent to solving Formulation~\ref{eq:traOD} with certain route choice probability $p$. Note the bound applies for all route choice probability $p$.  Therefore the statistical risk of the estimated O-D mean is still of ${\cal O}\left(\frac{1}{n}\right)$.
\end{proof}

\end{document}